\documentclass[12pt,a4paper]{article}
\usepackage[T1]{fontenc}
\usepackage{fancyhdr}
\usepackage{amsmath,amsthm,amsfonts,amssymb}
\allowdisplaybreaks
\usepackage{epic,eepic}
\usepackage{graphicx}
\usepackage{authblk}
\usepackage{hyperref}
\usepackage{fullpage}
\usepackage[active]{srcltx}
\usepackage{enumitem}
\usepackage{booktabs}
\usepackage{mathtools}

\usepackage{framed,color}

\usepackage{float}%
\floatstyle{plaintop}%
\restylefloat{table}

\usepackage[font={small,it}]{caption}

\usepackage[table]{xcolor}

\usepackage[toc,page]{appendix}
\graphicspath{{Figures/}}

\usepackage[style=authoryear,backend=bibtex,maxcitenames=2,maxbibnames=99]{biblatex}
\bibliography{LPM}

\usepackage{chngcntr}
\usepackage{dsfont}

\usepackage{algpseudocode}
\usepackage{algorithm}
\usepackage{setspace}

\usepackage{tikz}
\usetikzlibrary{bayesnet}
\usetikzlibrary{decorations.pathreplacing}

\makeatletter
\AtEveryBibitem{%
  \global\undef\bbx@lasthash%
  \clearfield{}}
\makeatother

\theoremstyle{plain}
\newtheorem{theorem}[]{Theorem}
\newtheorem{proposition}[]{Proposition}
\newtheorem{lemma}[]{Lemma}
\newtheorem{corollary}[]{Corollary}
\theoremstyle{definition}

\theoremstyle{remark}
\newtheorem{remark}{Remark}

\numberwithin{equation}{section}
\numberwithin{figure}{section}
\theoremstyle{plain}
\newtheorem{assumption}{Assumption}

\usepackage{color}
\def\eg{\textit{e.g.\,}}

\def\bpsi{\boldsymbol{\psi}}
\def\btheta{\boldsymbol{\theta}}
\def\Acal{\mathcal{A}}
\def\Scal{\mathcal{S}}

\def\uQ{\underline{Q}}
\def\bQ{\overline{Q}}
\def\upi{\underline{\pi}}
\def\bpi{\overline{\pi}}
\def\bz{\mathbf{z}}

\def\tP{\tilde{P}}
\def \rmd{\mathrm{d}}

\date{}
\title{\textbf{Computationally efficient inference for latent position network models}}
\author[1,*]{Riccardo Rastelli}
\author[2]{Florian Maire}
\author[1,3]{Nial Friel}
\affil[1]{\footnotesize School of Mathematics and Statistics, University College Dublin, Dublin, Ireland;}
\affil[2]{\footnotesize Department of Mathematics and Statistics, University of Montreal, Montreal, Canada;}
\affil[3]{\footnotesize Insight Centre for Data Analytics, Dublin, Ireland.\vspace{0.25cm}}
\affil[*]{\footnotesize riccardo.rastelli@ucd.ie}

\begin{document}
\rowcolors{2}{gray!25}{white}
\counterwithout{figure}{section}
\counterwithout{figure}{subsection}
\counterwithout{equation}{section}
\counterwithout{equation}{subsection}

\maketitle
\begin{abstract}
\noindent
Latent position models are widely used for the analysis of networks in a variety of research fields.
In fact, these models possess a number of desirable theoretical properties, and are particularly easy to interpret.
However, statistical methodologies to fit these models generally incur a computational cost which grows with the square of the number of nodes in the graph.
This makes the analysis of large social networks impractical.
In this paper, we propose a new method characterised by a much reduced computational complexity, which can be used to fit latent position models on networks of several tens of thousands nodes.
Our approach relies on an approximation of the likelihood function, where the amount of noise introduced by the approximation can be arbitrarily reduced at the expense of computational efficiency.
We establish several theoretical results that show how the likelihood error propagates to the invariant distribution of the Markov chain Monte Carlo sampler.
In particular, we demonstrate that one can achieve a substantial reduction in computing time and still obtain a good estimate of the latent structure.
Finally, we propose applications of our method to simulated networks and to a large coauthorships network, highlighting the usefulness of our approach.
\\

\end{abstract}

\baselineskip=20pt
\section{Introduction}\label{sec:intro}
In the last few decades, network data has become extremely common and readily available in a variety of fields, including the social sciences, biology, finance and technology.
After the pioneering work of \textcite{hoff2002latent}, latent position models (hereafter LPMs) have become one of the cornerstones in the statistical analysis of networks.
LPMs are flexible models capable of capturing many salient features of realised networks while providing results which can be easily interpreted.
However, a crucial aspect in the statistical analyses of networks is scalability: the computational burden required when fitting LPMs generally grows with the square of the number of nodes.
This seriously hinders their applicability, since estimation becomes impractical for networks larger than a few hundreds nodes.
Here, we precisely address this issue by introducing a new scalable methodology to fit LPMs: we study the new approach by providing theoretical guarantees on its efficiency, and we illustrate its use on simulated and real datasets.

LPMs postulate that the nodes of an observed network are characterised by a unique random position in a latent space: in the most common setup, each node is mapped to a point of $\mathbb{R}^2$. Additionally, the probability of observing an edge between two nodes is determined by the corresponding pairwise latent distance. A common assumption requires that closer nodes are more likely to connect than nodes farther apart, or, equivalently, that the probability of connection $\rho\left( d_{ij} \right)$ is a non-increasing function of the distance $d_{ij}$ between nodes $i$ and $j$. Evidently, the aforementioned quadratic computing costs originate from the necessity of keeping track of all of the pairwise distances between the nodes.

In our approach, we construct a partition of the latent space, therefore inducing a partition on the nodes of the graph itself. This allows us to cluster together nodes that are expected to have approximately the same behaviour, with regard to their connections. In principle, this is similar to imposing a stochastic block model structure \parencite{wang1987}, whereby the nodes belonging to the same block are assumed to be \textit{stochastically equivalent} \parencite{nowicki2001estimation}. The crucial advantage of our approach is that, once the partitioning has been set up, we can bypass the calculation of the pairwise distances via an approximation, in fact reducing to a computational complexity (in the number of nodes) that is lower than than of the standard methods.

Similarly to the original paper of \textcite{hoff2002latent}, our approach also relies on Markov chain Monte Carlo (hereafter MCMC) to obtain a Bayesian posterior sample of the latent positions and other model parameters.
However, in contrast to their approach, we replace the likelihood of the LPM with an approximate (hence \textit{noisy}) counterpart that aggregates the latent position of nodes belonging to the same block.
By construction, the cost of the calculation of this surrogate likelihood grows linearly in the number of nodes, hence giving a significant computational advantage to our method when compared to the approach of \textcite{hoff2002latent}
or other subsequent related works.

Since the LPM likelihood is replaced by a proxy, our method broadly fits within the context of \textit{noisy} Markov chain Monte Carlo \parencite{alquier2016noisy},
a topic that has recently generated a noticeable interest within the field of computational statistics and beyond.
The theoretical aspect of our paper relies and builds upon the core ideas of noisy MCMC.
In particular, our methodology is supported by a collection of theoretical results showing that our approach leads to quantifiable gains in efficiency.
More precisely, we show that the error in the MCMC output induced by the likelihood approximation can be arbitrarily bounded by refining the partition in the latent space.
Besides, a finer partition also implies higher computational costs.
As a consequence, our algorithm allows a trade-off between speed and accuracy that can be set according to the available computational budget, and the level of precision required for inference.
We study in detail how this approximation error is affected by the fineness of the partition as the number of nodes increases, hence providing a detailed characterisation of the trade-off.
In addition, our theoretical developments include a proposition that can be regarded as an extension of the results of \textcite{alquier2016noisy} to the widely used Metropolis-within-Gibbs (MwG hereafter) algorithm,
and which may thus have applications beyond the context of LPMs.
The theoretical results are established for a generic LPM framework: the assumptions we use are rather general and encompass most of the commonly used LPMs.

In addition to these results, we propose applications of our method to both simulated and real datasets, whereby we focus on a more specific model which is equivalent to that of \textcite{hoff2002latent}.
Our simulation study aims at assessing the approximation error bounds from a much more practical perspective, highlighting the validity of the procedure in asymptotic settings, and providing useful indications on how one should set up the partitioning.
The study demonstrates that the noisy algorithm succeeds in recovering the latent structure correctly, achieving the same qualitative results obtained with the currently available approaches.
Crucially, the computing time required by our proposed approach is only a fraction of that of the non-noisy one.

Finally, we propose an application to a large social network representing coauthorships in the astrophysics category of the repository of electronic preprints, arXiv.
This application demonstrates that our approach can be successfully employed on very large\footnote{This is meant in comparison to previous other analyses in the LPM literature.} networks, to recover the structure of the latent space while using just a small fraction of the actual computational cost.
This effectively extends the applicability of latent position models to much larger network datasets.

The structure of the paper is as follows: in Section \ref{sec:related_literature} we give an overview of the literature related to LPMs and noisy Markov chain Monte Carlo.
In Section \ref{sec:LPMs}, we formally characterise the main features of the original LPM of \textcite{hoff2002latent}, giving an overview of the MwG sampling strategy used to perform inference, highlighting some of its limitations.
In Section \ref{sec:Assumptions}, we lay the foundations for our theoretical results, by defining the general assumptions that our LPMs must satisfy.
In Section \ref{sec:GridStructure}, we formally introduce the partitioning of the latent space and all of the associated notation.
Section \ref{sec:noisy_mcmc} introduces the novel noisy algorithm, whereas in Section \ref{sec:theoretical_guarantees} we describe the main theoretical results. Finally, Sections \ref{sec:experiments} and \ref{sec:coauthorship} illustrate the applications of our methodology to simulated and a real dataset, respectively.

\section{Review of related literature}\label{sec:related_literature}
The study of the mathematical properties of LPMs dates back at least to \textcite{gilbert1961random}.
However, the first application of these models in the statistical analysis of social networks is due to \textcite{hoff2002latent}, who introduced a feasible methodology to fit LPMs to interaction data.
Since the work of \textcite{hoff2002latent}, LPMs have been intensively studied and widely applied to a variety of contexts, becoming one of the prominent statistical models for network analyses.
There are a number of reasons for this success.
Most importantly, LPMs are particularly easy to interpret, and offer a clear and intuitive graphical representation of the results.
In addition, LPMs are capable of capturing a number of features of interest such as transitivity, clustering, homophily and assortativity, which are often exhibited by observed social networks.
An overview of the theoretical properties of realised LPMs is given in \textcite{rastelli2015properties}.

In order to increase the flexibility of these models, a number of extensions of the basic framework have been considered.
\textcite{handcock2007model} introduce a more sophisticated prior on the latent point process to represent clustering in the network, that is, the presence of communities.
\textcite{krivitsky2009representing} further extends the model to include nodal random effects, i.e. additional latent features on the nodes capable of tuning their in-degrees and out-degrees.
Both of these extensions are implemented in the \texttt{R} package \texttt{latentnet}.

LPMs have also been extended to account for multiple network views \parencite{gollini2014joint,durante2017nonparametric,salter2017latent},
binary interactions evolving over time \parencite{sarkar2006dynamic, sewell2015latent, friel2016, durante2016locally},
ranking network data \parencite{gormley2007latent, sewell2015analysis} and weighted networks \parencite{sewell2016latent}.
Review papers dealing with LPMs include \textcite{salter2012review}, \textcite{matias2014modeling} and \textcite{raftery2017comment}.

Similarly to our contribution, three other papers address the issue of scalability for the inference on LPMs.
In \textcite{salter2013variational}, the authors propose a variational approximation (coupled with first order Taylor expansions to deal with various intractabilities) to perform posterior maximisation for the model described by \textcite{handcock2007model}.
One drawback of this approach is that it is not possible to assess the magnitude of the error induced by the variational approximation.
Also, the modelling assumptions are not flexible, since the variational framework can only be used with a restricted selection of parametric distributions.

In \textcite{ryan2017bayesian}, the authors consider the same latent position clustering model, and propose a Gaussian finite mixture prior distribution on the latent point process that allows one to \textit{collapse} the posterior distribution.
This means that several model parameters can be analytically integrated out from the posterior distribution of the model, hence simplifying the sampling scheme and achieving better estimators with a smaller computational cost.

Finally, \textcite{raftery2012fast} proposes a case-control likelihood approximation for the LPM with nodal random effects.
In this paper, the authors argue that the majority of large social networks are sparse, hence, missing edges contribute the most to the LPM likelihood.
By analogy with the case-control idea from epidemiology, they estimate the likelihood value using only a subset of the contributions given by the missing edges.
We consider this approach similar to ours, since both methods rely on a noisy likelihood.
We point out that our algorithm benefits from a series of theoretical results that guarantee its correctness and characterise the error induced by the approximation.
In addition, our method may be applied to networks of potentially huge size regardless of the level of sparseness.

Regarding the theoretical analysis of our algorithm, the main reference that we relate to is \textcite{alquier2016noisy}.
These authors argue that the computational problems arising when dealing with intractable likelihoods, or when inferring very large datasets, can often be alleviated by introducing approximations in the MCMC schemes.
These approaches are generally referred to as noisy MCMC, since one ends up sampling using a noisy transition kernel, rather than the correct one.
In \textcite{alquier2016noisy}, the authors exploit a theoretical result from \textcite{mitrophanov2005sensitivity} to characterise the error induced by these approximations on the invariant distribution of the transition kernel.
They also propose several applications based on the Metropolis-Hastings algorithm to a number of relevant statistical modelling frameworks.
We also point out that, more recently,
the noisy Monte Carlo framework has been adopted by \textcite{boland2017efficient} and \textcite{maire2017informed},
as a means to speed up inference for Gibbs random fields and other general models.
Even though the literature on noisy MCMC has been recently enriched by a number of relevant entries \parencite{negrea2017error,johndrow2017error,rudolf2017perturbation},
the theoretical framework developed in \textcite{alquier2016noisy} proved sufficient to establish our results, as shown in Section \ref{sec:theoretical_guarantees}.

\section{Latent Position Models}\label{sec:LPMs}
\subsection{Definition}\label{sec:LPM_definition}
A random graph is an object $\mathcal{G} = \left\{\mathcal{V},\mathcal{E}\right\}$ where $\mathcal{V} = \left\{1,\dots,N\right\}$ is a fixed set of labels for the nodes and
$\mathcal{E}$ is a list of the randomly realised edges.
In the social sciences, for example, random graphs are used to represent the social interactions within a set of actors.
The values appearing on the undirected ties are modeled through the random variables:
\begin{equation}
 \mathcal{Y} = \left\{ Y_{ij}: i,j\in\mathcal{V},\ i<j\right\}.
\end{equation}
In this paper we only deal with undirected binary graphs, hence, the observed realisations are denoted as follows:
\begin{equation}
\rowcolors{1}{}{}
 y_{ij}=\begin{cases}
  1,&\mbox{if an edge between $i$ and $j$ appears;}\\
  0,&\mbox{otherwise;}
 \end{cases}
\end{equation}
for every $i\in\mathcal{V}$ and $j\in\mathcal{V}$ such that $j>i$.
Note that, in the framework considered, self-edges are not modelled.

In LPMs the nodes are characterised by a latent position, generically denoted $\textbf{z}\in\mathbb{R}^m$, which determines their social profile.
The choice $m=2$ is the most common since it usually couples a good fit and a convenient framework to represent the results.
Hence, we illustrate our methodology assuming that the number of latent dimensions is two, noting that extensions to other cases may be possible.

In the basic LPM, the probability of an edge appearing is determined by the positions of the nodes at its extremes and by some other global parameters (e.g. an intercept).
This may be formally written as follows:
\begin{equation}
 p\left( \textbf{z}_i, \textbf{z}_j; \boldsymbol{\psi} \right) := \mathbb{P}\left( y_{ij} = 1 \middle\vert \textbf{z}_i, \textbf{z}_j, \boldsymbol{\psi}\right) = 1 - \mathbb{P}\left( y_{ij} = 0 \middle\vert \textbf{z}_i, \textbf{z}_j, \boldsymbol{\psi} \right).
\end{equation}
Here $\bpsi$ is a vector of global parameters with dimensions indexed by the labels $\mathcal{K} = \left\{ 1,\dots,K\right\}$. The parameter $\bpsi$ is sometimes referred to as the static parameter of the model, as opposed to the latent field $\mathcal{Z}:=\{\bz_1,\ldots,\bz_N\}$.
A number of possible formulations for the edge probabilities have been proposed.
Within the statistical community, the most common choice is the logit link proposed by \textcite{hoff2002latent}:
\begin{equation}
\label{eq:logit_link}
 \log\left( \frac{p\left( \textbf{z}_i, \textbf{z}_j; \psi \right)}{1-p\left( \textbf{z}_i, \textbf{z}_j; \psi \right)} \right) := \psi - d\left( \textbf{z}_i, \textbf{z}_j \right);
\end{equation}
where $d\left( \textbf{z}_i, \textbf{z}_j \right)$ denotes the Euclidean distance between the two nodes, and $\psi\in\mathbb{R}$ is simply an intercept parameter ($K=1$).
Alternative formulations are used in \textcite{gollini2014joint} and \textcite{rastelli2015properties}.
In physics, a variety of edge probability functions have been proposed.
A list of these can be found, for example, in \textcite{parsonage2015fast} and references therein.
One feature that all of these formulations have in common is that the edge probability is a function of the distance between the two nodes, and that its value decreases as the latent distance increases, making long edges less likely to appear.

Since the data observations are conditionally independent given the latent positions, the likelihood of all undirected LPMs may be written as:
\begin{equation}\label{eq:likelihood_1}
 \mathcal{L}_{\mathcal{Y}}\left( \mathcal{Z},\boldsymbol{\psi} \right) = \mathbb{P}\left( \mathcal{Y} \middle\vert \mathcal{Z}, \boldsymbol{\psi}\right) =
 \prod_{\left\{i\in\mathcal{V}\right\}} \prod_{\left\{j\in\mathcal{V}\setminus{i}\right\}}
 \left\{\left[p\left( \textbf{z}_i, \textbf{z}_j; \boldsymbol{\psi} \right)\right]^{y_{ij}}\left[1-p\left( \textbf{z}_i, \textbf{z}_j; \boldsymbol{\psi} \right)\right]^{1-y_{ij}}\right\}^{1/2}
\end{equation}
where the square root is introduced to remedy the fact that each edge contributes twice to the likelihood of the undirected network (the motivation behind this particular formulation will be more clear in the following sections).
We note that, for a given set of positions $\mathcal{Z}$ and global parameters $\boldsymbol{\psi}$, the computational cost for the likelihood evaluation is $\mathcal{O}\left( N^2 \right)$, i.e. it grows with the square of the number of nodes.

\subsection{Bayesian inference}\label{sec:LPM_Bayesian_Inference}
Inference for LPMs is usually carried out in a Bayesian framework, using MCMC to obtain posterior samples of the model parameters \parencite{hoff2002latent, handcock2007model, krivitsky2009representing, raftery2012fast}.
The posterior distribution of interest is:
\begin{equation}\label{eq:posterior1}
 \pi\left( \mathcal{Z}, \bpsi \middle\vert \mathcal{Y}\right) \propto \mathcal{L}_{\mathcal{Y}}\left( \mathcal{Z}, \bpsi \right) \pi\left( \mathcal{Z} \right) \pi\left( \bpsi \right).
\end{equation}
Assuming that the cost of the evaluation of the priors $\pi\left( \mathcal{Z} \right)$ and $\pi\left( \bpsi \right)$ is $\mathcal{O}\left( N \right)$ or negligible,
the computational cost required to evaluate the posterior value grows with $N^2$, which corresponds to the bottleneck imposed by the likelihood term.
A Markov chain Monte Carlo sampler can be designed to sample each of the model parameters in turn, using the following full-conditional distributions:
\begin{equation}\label{eq:fullconditional_z}
 \pi\left( \textbf{z}_i \middle \vert \mathcal{Z}_{-i},\bpsi,\mathcal{Y} \right) \propto \pi\left( \textbf{z}_i \right)
 \prod_{\left\{j\in\mathcal{V}:\ j\neq i\right\}} \left[p\left( \textbf{z}_i, \textbf{z}_j;\bpsi\right)\right]^{y_{ij}}\left[ 1-p\left( \textbf{z}_i, \textbf{z}_j;\bpsi\right) \right]^{1-y_{ij}}
\end{equation}
\begin{equation}\label{eq:fullconditional_psi}
 \pi\left( \psi_k \middle \vert \bpsi_{-k},\mathcal{Z},\mathcal{Y} \right) \propto \pi\left( \psi_k \right) \mathcal{L}_{\mathcal{Y}}\left( \mathcal{Z}, \bpsi \right)
\end{equation}
In the previous equations: $i\in\mathcal{V}$, $k\in \mathcal{K}$, whereas $\mathcal{Z}_{-i} = \left\{\textbf{z}_j\right\}_{j\in\mathcal{V}\setminus\{i\}}$ and $\bpsi_{-k} = \left\{\psi_{k'}\right\}_{k'\in\mathcal{K}\setminus\{k\}}$.
Here we have assumed that the model parameters are all independent a priori: this is indeed very common and it will be formalised in the following sections.
Each evaluation of \eqref{eq:fullconditional_psi} clearly requires $\mathcal{O}\left( N^2 \right)$.
Since each evaluation of \eqref{eq:fullconditional_z} requires $\mathcal{O}\left( N \right)$ calculations, the overall complexity of the sampler still grows with the square of $N$.

The full-conditionals \eqref{eq:fullconditional_z} and \eqref{eq:fullconditional_psi} are generally not in standard form.
Hence, new values for the model parameters are sampled through what is usually referred to as a Metropolis-within-Gibbs (MwG) type algorithm (see \eg \cite{gilks1995adaptive}). More precisely, potential new parameters are drawn from proposal distributions $q_{\mathcal{Z}}\left( \textbf{z}_i\rightarrow \textbf{z}_i'\right)$ and $q_{\bpsi}\left( \psi_k\rightarrow \psi_k'\right)$ and are then accepted with probability:
\begin{equation}\label{eq:acceptance_z}
 \alpha_{\mathcal{Z}}\left( \textbf{z}_i\rightarrow \textbf{z}_i'\right) := 1\wedge\left\{ \frac{q_{\mathcal{Z}}\left( \textbf{z}_i'\rightarrow \textbf{z}_i\right) \pi\left( \textbf{z}_i' \middle \vert \mathcal{Z}_{-i},\bpsi,\mathcal{Y} \right)}
 {q_{\mathcal{Z}}\left( \textbf{z}_i\rightarrow \textbf{z}_i'\right)\pi\left( \textbf{z}_i \middle \vert \mathcal{Z}_{-i},\bpsi,\mathcal{Y} \right)}\right\}
\end{equation}
\begin{equation}\label{eq:acceptance_psi}
 \alpha_{\bpsi}\left( \psi_k\rightarrow \psi_k'\right) := 1\wedge\left\{ \frac{q_{\bpsi}\left( \psi_k'\rightarrow \psi_k\right) \pi\left( \psi_k' \middle \vert \bpsi_{-k},\mathcal{Z},\mathcal{Y} \right)}
 {q_{\bpsi}\left( \psi_k\rightarrow \psi_k'\right)\pi\left( \psi_k \middle \vert \bpsi_{-k},\mathcal{Z},\mathcal{Y} \right)}\right\}
\end{equation}
for the latent positions and global parameters, respectively.
In the previous equations, for two real numbers $a$ and $b$, $a\wedge b$ stands for the minimum between the two numbers.
Also, we point out that, as is common practice, the two dimensions of the latent positions are dealt with simultaneously, i.e. they are updated in block.

The MwG sampler described above defines a Markov chain whose stationary distribution is the posterior of interest \eqref{eq:posterior1}. As a consequence, provided that the Markov chain is ergodic, the samples obtained at stationarity can be used to fully characterise the posterior distribution of interest. In fact, the MwG chain is shown to be geometrically ergodic for a variety of proposal distributions and under some regulatory conditions on the invariant distribution $\pi$, see \cite[Theorem 5]{roberts1998two}.

\subsection{Non-identifiability of the latent positions}\label{sec:LPM_unidentifiable}
LPMs are known to be non-identifiable with respect to translations, rotations, and reflections of the latent positions.
This issue has no particular effect on the sampling itself, yet it may hinder the interpretation of the posterior samples.
For this reason, the latent positions are usually post-processed using the so-called Procrustes' matching.
This procedure consists of rotating
and translating
the configurations of points observed at the end of each iteration, to match a given reference layout.
In this way, the trajectory of each node during the sampling may be properly assessed, since the overall rotation
and translation
effect has been removed.
A detailed description of the method is given, for example, in \textcite{hoff2002latent} and \textcite{shortreed2006positional}.
In this paper, we adopt exactly this same strategy to solve the non-identifiability problem, using as reference either the true positions (if available) or the maximum a posteriori configuration.

\section{Assumptions}\label{sec:Assumptions}
The methodology we develop in this paper relies on several assumptions which are described in this section.

\begin{assumption}\label{assumption_bounded_support}
All of the model parameters are defined on bounded sets, i.e.:
\begin{equation}
 \forall\, k\in\mathcal{K}: \psi_{k} \in [\psi_k^{\mathcal{L}},\psi_k^{\mathcal{U}}] =: \mathcal{S}_{k},
\end{equation}
\begin{equation}
 \forall\, i\in\mathcal{V}: \textbf{z}_i \in [-S,S]\times[-S,S] = :\mathcal{S}_{\mathcal{Z}},
\end{equation}
for some finite positive constants $S$, $\psi_k^{\mathcal{L}}$ and $\psi_k^{\mathcal{U}}$.
\end{assumption}

\begin{remark}
We note that, as a consequence of Assumption \ref{assumption_bounded_support}, the parameter space:
\begin{equation}
 \mathcal{S} = \mathcal{S}_{\mathcal{Z}}^{N} \times \mathcal{S}_{\{k=1\}} \times \cdots \times \mathcal{S}_{\{k=K\}}
\end{equation} 
is a compact set.
\end{remark}
\begin{remark}
 Assumption \ref{assumption_bounded_support} is rather strong and  contrasts with the usual LPM frameworks.
 However, we argue that, from a practical point of view, these imposed conditions do not change the essence of the model.
 In fact, very large LPM parameters normally lead to degenerate models, and hence to realised networks that are meaningless in this modelling context (e.g. full or empty graphs).
 In this perspective, there is in fact a necessity to constrain $\mathbf{\psi}$ to a bounded space in order to make the model more tractable.
\end{remark}


\begin{remark}
In the applications sections of this paper, a spherical truncated Gaussian distribution is used as prior on the latent positions:
\begin{equation}\label{eq:prior_z_1}
 \pi\left( \textbf{z}_i \right) = \prod_{m=1}^2 \left\{\frac{\phi\left( \frac{z_{im}}{\gamma} \right)}{\gamma\left[ \Phi\left( \frac{S}{\gamma} \right) - \Phi\left( \frac{-S}{\gamma} \right)\right]}\right\},\hspace{2cm} \forall i\in\mathcal{V};
\end{equation}
where $z_{im}$ indicates the $m$-th coordinate of node $i$'s latent position, $\gamma > 0$, $\phi$ and $\Phi$ are the p.d.f. and c.d.f. of a standard Gaussian distribution, respectively.
This prior specification is essential in order to match the grid construction that is described in Section \ref{sec:GridStructure}, however other priors could be considered.
\end{remark}

\begin{assumption}\label{assumption_edge_probability}
 The edge probability function $p:\mathbb{R}^2\times \mathbb{R}^2 \times \mathbb{R}^K \rightarrow [p^{\mathcal{L}},p^{\mathcal{U}}] \subset \left( 0,1 \right)$ satisfies the following properties:
 \begin{enumerate}
  \item[\textbf{a)}] $p$ \textbf{depends on the positions only through the latent distances}. This means that there exists a function $\rho:\mathbb{R}^+\times \mathbb{R}^K \rightarrow [p^{\mathcal{L}},p^{\mathcal{U}}]$ such that
  $$
  \forall \textbf{z}_i,\textbf{z}_j\in\mathbb{R}^2,\ \forall \bpsi\in\mathbb{R}^K:\ p\left(\textbf{z}_i,\textbf{z}_j;\bpsi\right) = \rho\left( d\left(\textbf{z}_i,\textbf{z}_j\right), \bpsi \right).
  $$
  \item[\textbf{b)}] $\rho$ \textbf{is non-increasing w.r.t. distances}; i.e. for any $\textbf{z}_i\in\mathcal{S}_{\mathcal{Z}}$, $i=1,2,3,4$:
  $$
  \mbox{ if } d\left(\textbf{z}_1,\textbf{z}_2\right) \geq d\left(\textbf{z}_3,\textbf{z}_4\right), \mbox{ then } p\left(\textbf{z}_1,\textbf{z}_2;\bpsi\right) \leq p\left(\textbf{z}_3,\textbf{z}_4;\bpsi\right).
  $$
  \item[\textbf{c)}] $\rho$ \textbf{is Lipschitz w.r.t. distances}; i.e. for any $\textbf{z}_i\in\mathcal{S}_{\mathcal{Z}}$, $i=1,2,3,4$:
  $$
  \left|p\left(\textbf{z}_1,\textbf{z}_2;\bpsi\right) - p\left(\textbf{z}_3,\textbf{z}_4;\bpsi\right)\right| \leq \kappa \left|d\left(\textbf{z}_1,\textbf{z}_2\right) - d\left(\textbf{z}_3,\textbf{z}_4\right)\right| \ ;
  $$
  for some finite positive constant $\kappa$.
 \end{enumerate}
\end{assumption}

\begin{remark}
Assumption \ref{assumption_edge_probability} is satisfied by most link functions, including the logit link of Eq. \eqref{eq:logit_link}.
\end{remark}

%

\section{Grid approximation of the latent distances}\label{sec:GridStructure}
Hereafter, we consider a generic LPM satisfying Assumptions \ref{assumption_bounded_support} and \ref{assumption_edge_probability},
and we illustrate an estimation procedure based on a grid partitioning of the latent space.
Following an approach similar to that of \textcite{parsonage2015fast}, we create a partitioning of the latent positions $\mathcal{Z}$ using a grid in $\mathbb{R}^2$.
The grid is made of adjacent squares (called boxes hereafter) of side length $b > 0$, each having both sides aligned to the axes.
A generic box $B[g,h]$ has corners located in $\left( bg-b,bh-b \right),\left( bg-b,bh \right),\left( bg,bh \right)$ and $\left( bg,bh-b \right)$,
where the indexes $g$ and $h$ are positive or negative but non-null integers, i.e. $g,h\in\mathbb{Z}\setminus{0}$.
Figure \ref{fig:grid1} shows the latent space with the partitioning given by these boxes.

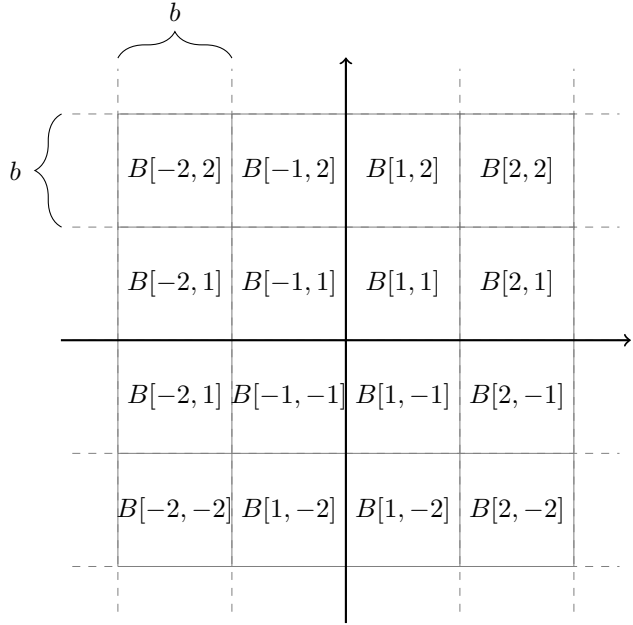
\begin{figure}[htbp]
\begin{center}
\begin{tikzpicture}[scale=1.5]
  \draw[step=1cm,gray,dashed,very thin] (-2.4,-2.4) grid (2.4,2.4);
  \draw[step=1cm,gray,very thin] (-2.0,-2.0) grid (2.0,2.0);
    \draw [<->,thick] (0,2.5) node (yaxis) [above] {}
        |- (2.5,0) node (xaxis) [right] {};
    \draw [-,thick] (0,-2.5) node (yaxis) [above] {}
        |- (-2.5,0) node (xaxis) [right] {};
  \node[obs, draw=none, fill=none, minimum size=0.5cm]  (z1) at (0.5,0.5) {$B[1,1]$};
  \node[obs, draw=none, fill=none, minimum size=0.5cm]  (z1) at (1.5,0.5) {$B[2,1]$};
  \node[obs, draw=none, fill=none, minimum size=0.5cm]  (z1) at (1.5,1.5) {$B[2,2]$};
  \node[obs, draw=none, fill=none, minimum size=0.5cm]  (z1) at (0.5,1.5) {$B[1,2]$};
  \node[obs, draw=none, fill=none, minimum size=0.5cm]  (z1) at (-0.5,-0.5) {$B[-1,-1]$};
  \node[obs, draw=none, fill=none, minimum size=0.5cm]  (z1) at (-1.5,-0.5) {$B[-2,1]$};
  \node[obs, draw=none, fill=none, minimum size=0.5cm]  (z1) at (-0.5,-1.5) {$B[1,-2]$};
  \node[obs, draw=none, fill=none, minimum size=0.5cm]  (z1) at (-1.5,-1.5) {$B[-2,-2]$};
  \node[obs, draw=none, fill=none, minimum size=0.5cm]  (z1) at (-0.5,0.5) {$B[-1,1]$};
  \node[obs, draw=none, fill=none, minimum size=0.5cm]  (z1) at (-1.5,0.5) {$B[-2,1]$};
  \node[obs, draw=none, fill=none, minimum size=0.5cm]  (z1) at (-0.5,1.5) {$B[-1,2]$};
  \node[obs, draw=none, fill=none, minimum size=0.5cm]  (z1) at (-1.5,1.5) {$B[-2,2]$};
  \node[obs, draw=none, fill=none, minimum size=0.5cm]  (z1) at (0.5,-0.5) {$B[1,-1]$};
  \node[obs, draw=none, fill=none, minimum size=0.5cm]  (z1) at (1.5,-0.5) {$B[2,-1]$};
  \node[obs, draw=none, fill=none, minimum size=0.5cm]  (z1) at (0.5,-1.5) {$B[1,-2]$};
  \node[obs, draw=none, fill=none, minimum size=0.5cm]  (z1) at (1.5,-1.5) {$B[2,-2]$};

  \draw [decorate,decoration={brace,amplitude=10pt,raise=4pt},yshift=0pt] (-2.4,1) -- (-2.4,2) node [black,midway,xshift=-0.75cm] {\footnotesize $b$};
  \draw [decorate,decoration={brace,amplitude=10pt,raise=4pt},yshift=0pt] (-2,2.4) -- (-1,2.4) node [black,midway,yshift=0.75cm] {\footnotesize $b$};

\end{tikzpicture}
\end{center}
\caption{Grid partitioning the latent space.}
\label{fig:grid1}
\end{figure}

We denote with $N[g,h]$ the number of points located in a generic box:
\begin{equation}
 N[g,h] = \left|\left\{ i\in\mathcal{V}:\textbf{z}_i\in B[g,h]\right\}\right|,
\end{equation}
where $\left|H\right|$ denotes the cardinality of the set $H$.

It is also useful to introduce the centre of a generic box $\textbf{c}[g,h] := \left( bg-b/2,bh-b/2 \right)$.
Given a node $j\in\mathcal{V}$ such that $\textbf{z}_j \in B[g,h]$, we also indicate the centre of $B[g,h]$ with $\textbf{c}_j$, representing the centre of the box containing $j$.
An essential aspect of our proposed approach is determined by the fact that the distance $d\left( \textbf{z}_i, \textbf{z}_j\right)$ between any two nodes may be approximated by $d\left( \textbf{z}_i, \textbf{c}_j\right)$, i.e. the distance between node $i$ and the centre of the box containing $j$.

Finally, we denote with $\xi_i[g,h]$ the number of edges between node $i$ and the nodes allocated to $B[g,h]$, i.e.:
\begin{equation}
 \xi_i[g,h] = \sum_{\left\{ j\in\mathcal{V}:\ \textbf{z}_j\in B[g,h]\right\}} y_{ij};
\end{equation}
and by $\zeta_i[g,h]$ the number of missing edges:
\begin{equation}
 \zeta_{i}[g,h] = N[g,h] - \xi_i[g,h] - \mathds{1}\left(\left\{z_i\in B[g,h]\right\}\right);
\end{equation}
where $\mathds{1}\left( \mathcal{A} \right)$ is $1$ if $\mathcal{A}$ is true or $0$ otherwise.
Also, the degree of node $i\in\mathcal{V}$, i.e. the number of edges incident to it, is indicated by $D_i$.

These quantities introduced are exploited in the following sections to illustrate a new way of carrying out Bayesian inference for LPMs, requiring a dramatically reduced computational cost.  

\section{Noisy MCMC}\label{sec:noisy_mcmc}
As explained in the previous section, the distance from node $i$ to the centre of a generic box $\textbf{c}[g,h]$ can be used as a proxy for the true distances between $i$ and all of the points contained in $B[g,h]$, for all $g$ and $h$.
This in turn allows one to approximate the edge probability $p\left( \textbf{z}_i, \textbf{z}_j;\bpsi\right)$ using $p\left( \textbf{z}_i, \textbf{c}_j;\bpsi\right)$, for all $j\in \mathcal{V}$ such that $\textbf{z}_j\in B[g,h]$.
This fact may be exploited in a number of ways. For example, the likelihood defined in \eqref{eq:likelihood_1} may be replaced by the following noisy likelihood:
\begin{equation}\label{eq:noisy_likelihood_1}
 \tilde{\mathcal{L}}_{\mathcal{Y}}\left( \mathcal{Z},\boldsymbol{\psi} \right) := \left\{\prod_{i=1}^{N}\prod_{g,h} \left[p\left( \textbf{z}_i, \textbf{c}[g,h];\boldsymbol{\psi}\right)\right]^{\xi_i[g,h]}\left[ 1-p\left( \textbf{z}_i, \textbf{c}[g,h];\boldsymbol{\psi}\right) \right]^{\zeta_i[g,h]}\right\}^{1/2};
\end{equation}
where each edge contribution is essentially replaced by its noisy counterpart.
Here, by counting each edge contribution twice and then correcting with the square root, one has the possibility to use the noisy approximation in a symmetric way, with respect to  any pair of nodes $i$ and $j$.
We point out that a number of alternative estimators are available for the likelihood value using the grid approximation:
the estimator proposed in \eqref{eq:noisy_likelihood_1} is one that generally works well in practice and that makes our theoretical developments easier to follow.

With \texttt{NoisyLPM}, we refer to a MwG sampler that relies on the approximate edge probabilities rather than the true ones, or, equivalently,
that uses the noisy likelihood $\tilde{\mathcal{L}}_{\mathcal{Y}}$ instead of the true likelihood ${\mathcal{L}}_{\mathcal{Y}}$.
In \texttt{NoisyLPM} the full-conditionals introduced in \eqref{eq:fullconditional_z} and \eqref{eq:fullconditional_psi} can be approximated as follows:
\begin{equation}\label{eq:fullconditional_z_noisy}
 \tilde{\pi} \left( \textbf{z}_i \middle \vert \mathcal{Z}_{-i},\bpsi,\mathcal{Y} \right) \propto \pi\left( \textbf{z}_i \right)
 \prod_{g,h} \left[p\left( \textbf{z}_i, \textbf{c}[g,h];\bpsi\right)\right]^{\xi_i[g,h]}\left[ 1-p\left( \textbf{z}_i, \textbf{c}[g,h];\bpsi\right) \right]^{\zeta_i[g,h]}\ ;
\end{equation}
\begin{equation}\label{eq:fullconditional_psi_noisy}
 \tilde{\pi} \left( \psi_k \middle \vert \bpsi_{-k},\mathcal{Z},\mathcal{Y} \right) \propto \pi\left( \psi_k \right) \tilde{\mathcal{L}}_{\mathcal{Y}}\left( \mathcal{Z}, \bpsi \right)\ .
\end{equation}

It is apparent that the computational cost of one evaluation of the approximate full-conditionals is much smaller than that of the true counterpart.
In fact, for a given grid, the complexity of a noisy MwG update becomes $\mathcal{O}\left( 1 \right)$ and $\mathcal{O}\left( N \right)$ for latent positions and global parameters, respectively.
Overall, this makes the computational complexity of the \texttt{NoisyLPM} procedure of an order smaller than $\mathcal{O}\left( N^2 \right)$.

\section{Theoretical guarantees}\label{sec:theoretical_guarantees}
This section provides theoretical results that characterise the error induced by our approximation.
Indeed, replacing $\mathcal{L}_{\mathcal{Y}}(\mathcal{Z},\boldsymbol{\psi})$ with $\tilde{\mathcal{L}}_{\mathcal{Y}}(\mathcal{Z},\boldsymbol{\psi})$
in the MwG acceptance ratio implies that the stationary distribution of the Markov chain may not coincide anymore with the posterior distribution of interest described in Section \ref{sec:LPM_Bayesian_Inference}.
Here, our main goal is to show that a noisy MwG sampler, such as the \texttt{NoisyLPM}, generates a sequence of random variables whose distribution can be made arbitrarily close to the true posterior $\pi\left(\ \cdot\ \middle\vert \mathcal{Y} \right)$.

In fact, one can note that, by construction, our noisy MwG sampler admits the approximate posterior $\tilde{\pi}\left(\ \cdot\ \middle\vert \mathcal{Y} \right)$ as stationary distribution.
Hence, the approximation error is directly, and globally, measured by $\|\pi-\tilde{\pi}\|$, i.e. the total variation distance between the two posteriors.
However, obtaining an explicit expression or an upper bound of $\|\pi-\tilde\pi\|$ is challenging.
Our main result (Theorem \ref{thm:main}) gives an upper bound of $\|\pi-\tilde\pi\|$ obtained by bounding the distance between the exact and noisy Markov chains, following \textcite{mitrophanov2005sensitivity}. The core of our work has been to devise a bound, which we believe is tight, on the distance between the two Markov kernels, see Theorem \ref{theorem_acceptance_probs_1} and Corollary \ref{corollary_2}.


The theoretical framework is the analysis of the perturbation of uniformly ergodic Markov chains, initiated in \textcite{mitrophanov2005sensitivity} and refined for the noisy Metropolis-Hastings case in \textcite{alquier2016noisy}. We first recall the uniform ergodicity assumption.
\begin{assumption}
\label{assumption_uniform}
A $\pi$-invariant Markov kernel $P$ operating on a state space $\mathcal{S}$ is uniformly ergodic if after $t\in\mathbb{N}$ iterations, the distance between the chain distribution and the stationary distribution is bounded as follows:
 \begin{equation}
  \sup_{u\in\mathcal{S}}\|P^{t}(\btheta,\,\cdot\,) - \pi\| \leq C\tau^t,
 \end{equation}
for some $C<\infty$ and $\tau<1$.
\end{assumption}

The section is divided in two parts: in the spirit of \textcite{alquier2016noisy}, we first derive an extension of their theoretical framework to include the analysis of noisy Metropolis-within-Gibbs algorithms in a generic setup,
that is, beyond the LPM context. In the second part we give a series of theoretical results that are specific to LPMs, and that aim to characterise the magnitude of the approximation error in the MwG acceptance probabilities, in preparation for applying our general result. In particular, we show that the distance between the exact algorithm and the \texttt{NoisyLPM} can be arbitrarily reduced by refining the latent grid.


\subsection{Noisy MwG aggregated errors}
\label{sec:theoretical_guarantees_sub_2}
This paper deals with an approximation of a MwG Markov chain, where the parameters of the model are updated in turn. Perturbations of uniformly ergodic Metropolis-Hastings Markov chains have been studied in \textcite{alquier2016noisy}.
We show, here, that a similar analysis can be carried out in a generic MwG sampler framework.


We introduce the following notation.
We indicate with $r$ a generic parameter update step of the MwG, with $R$ indicating the number of updates performed in a particular algorithmic instance.
For example, $R$ may indicate the number of model parameters that are updated in each iteration of the MCMC algorithm.

An arbitrary sigma-algebra on the compact parameter space $\mathcal{S}$ is denoted by $\mathcal{A}$.
For any signed measure $\mu$ on $(\mathcal{S},\mathcal{A})$, we denote the total variation distance of $\mu$ by $\|\mu\|:=\sup_{A\in\mathcal{A}}|\mu(A)|$.
For any Markov kernel $P$ taking values in $\mathcal{S}\times\mathcal{A}$, we denote the operator norm of $P$ as:
\begin{equation}\label{eq:norm}
 \|P\|: = \sup_{\btheta\in\mathcal{S}}\|P(\btheta,\cdot)\| = \sup_{\btheta\in\mathcal{S}}\sup_{A\in\mathcal{S}}\left|P(\btheta,A)\right|\,.
\end{equation}
Finally, let $\mu P$ be the measure on $(\mathcal{S},\mathcal{A})$ defined as $\mu P:=\int_{\mathcal{S}}\mu(\mathrm{d} x)P(x,\,\cdot\,)$.
The kernel $P$ will generally be considered as the \textit{exact} kernel, whereas $\tilde{P}$ will denote its noisy counterpart.

The following proposition shows that the distance between the one step transition of an elementary MwG update and its noisy counterpart is uniformly bounded.
\begin{proposition}\label{prop_noisy_corollary}
Let $\alpha\left( \boldsymbol{\theta} \rightarrow \boldsymbol{\theta}' \right)$ and $\tilde{\alpha}\left( \boldsymbol{\theta} \rightarrow \boldsymbol{\theta}' \right)$ be the corresponding exact and noisy acceptance probabilities (respectively), that arise when considering a generic update $\boldsymbol{\theta} \rightarrow \boldsymbol{\theta}'$.
If there exists some finite constant $\omega>0$ such that:
\begin{equation}
 \left| \alpha\left( \btheta \rightarrow \btheta' \right) - \tilde{\alpha}\left( \btheta \rightarrow \btheta' \right)\right| \leq \omega
\end{equation}
then we also have:
\begin{equation}
 \|P - \tilde{P}\| \leq \omega \ .
\end{equation}
\end{proposition}
The proof of this proposition is given in Appendix \ref{app:prop_noisy_corollary}.
Now, we characterize instead the error that is accumulated over a sweep of the MwG sampler over a collection of the model parameters.
We denote with $P_{[R]}$ (resp. $\tP_{[R]}$) the kernel corresponding to a sequential update of a number of model parameters using exact (resp. approximate) acceptance probability:
\begin{equation}\label{eq:composite_kernel}
\begin{split}
 P_{[R]}\left( \btheta,\cdot \right) &:= P_{1}\cdots P_{R}\left( \btheta,\cdot \right)\,, \\
 &= \int\cdots\int P_1\left( \btheta, d\btheta_1 \right)\cdots P_{R-1}\left( \btheta_{R-2},d\btheta_{R-1} \right)P_R\left( \btheta_{R-1}, \cdot \right)\,,\\
 \tP_{[R]}\left( \btheta,\cdot \right) &:= \tP_{1}\cdots \tP_{R}\left( \btheta,\cdot \right)\,.
\end{split}
\end{equation}
This corresponds to the composition of the $R$ elementary kernels, indicated with $P_r$ or $\tilde{P}_r$, each characterising the update of one model parameter.
\begin{proposition}\label{prop_subadditive}
The error carried by a product of noisy Markov kernels $\tilde{P}_{[R]} = \tilde{P}_1\tilde{P}_2\cdots\tilde{P}_R$ relative to its exact version $P_{[R]} = P_1P_2\cdots P_{R}$ is subadditive, in the sense that:
\begin{equation}\label{eq:thm_scan_1}
 \|P_{[R]} - \tilde{P}_{[R]}\| \leq \sum_{r=1}^{R} \|P_r - \tilde{P}_r\|\ .
\end{equation}
\end{proposition}
The proof of this proposition is provided in Appendix \ref{app:prop_subadditive}. We can join the above two results in the following proposition.

\begin{proposition}\label{prop_scan}
 Let $\alpha$ and $\tilde{\alpha}$ be the corresponding exact and noisy acceptance probabilities (respectively), that arise when considering a generic update for any of the model parameters. Assume that there exists some finite constant $\omega>0$, such that:
\begin{equation}
 \left| \alpha - \tilde{\alpha}\right| \leq \omega \ .
\end{equation}
Then, after $R < \infty$ parameter updates, the product kernels satisfy:
\begin{equation}
 \|P_{[R]} - \tilde{P}_{[R]}\| \leq R \omega \ .
\end{equation}
\end{proposition}
\begin{proof}
 This immediately follows from Propositions \ref{prop_noisy_corollary} and \ref{prop_subadditive}.
\end{proof}

Finally, as in \textcite{alquier2016noisy}, we rely on Corollary $3.1$ of \textcite{mitrophanov2005sensitivity} to give our main result for the \texttt{NoisyLPM} algorithm.
\begin{corollary}\label{corollary_mitrophanov}
 Let $P_{[R]}$ (resp. $\tilde{P}_{[R]}$) be the transition kernel for the exact MwG sampler (resp. noisy) described in Eq. \eqref{eq:composite_kernel}.  Assume that the Markov chain with kernel $P$ is uniformly ergodic (Assumption \ref{assumption_uniform}). Then, for any $t>0$ and for any starting point $\btheta\in\mathcal{S}$:
\begin{equation}
 \|\delta_{\btheta} P_{[R]}^t - \delta_{\btheta} \tilde{P}_{[R]}^t\| \leq \left( \lambda + \frac{C\tau^\lambda}{1-\tau} \right)R\omega \ ,
\end{equation}
where $\lambda = \lceil{\log\left( 1/C \right)}\slash{\log(\tau)}\rceil$.
\end{corollary}

\subsection{LPM likelihood errors}\label{sec:theoretical_guarantees_sub_1}
We now report theoretical results that are specific to LPMs, in preparation of applying Corollary \ref{corollary_mitrophanov} to this context.
In the following theorem, we show that the error on the MwG acceptance probabilities is bounded, and that it can be arbitrarily reduced by refining the latent grid partition.
\begin{theorem}\label{theorem_acceptance_probs_1}
 Under Assumptions \ref{assumption_bounded_support} and \ref{assumption_edge_probability}, the error on the acceptance probabilities for a latent position update satisfies for all $i\in\mathcal{V}$:
\begin{equation}\label{eq:theorem_noisy_alpha_z}
 \left| \alpha_{\mathcal{Z}}\left( \textbf{z}_i\rightarrow \textbf{z}_i'\right) - \tilde{\alpha}_{\mathcal{Z}}\left( \textbf{z}_i\rightarrow \textbf{z}_i'\right)\right|
 \leq \kappa' b N \ ,
 \end{equation}
 and for any $k\in\mathcal{K}$ identifying a static parameter's update:
\begin{equation}\label{eq:theorem_noisy_alpha_psi}
 \left| \alpha_{\bpsi}\left( \psi_k\rightarrow \psi_k'\right) - \tilde{\alpha}_{\bpsi}\left( \psi_k\rightarrow \psi_k'\right)\right|
 \leq \kappa'' b N^2 \ ,
\end{equation}
where $\kappa'$ and $\kappa''$ are suitable positive finite constants which do not depend on either $b$ or $N$.
\end{theorem}
The proof of Theorem \ref{theorem_acceptance_probs_1} is given in Appendix \ref{app:theorem_acceptance_probs_1}.

\begin{corollary}\label{corollary_2}
 Let $P_{[N+K]}$ and $\tilde{P}_{[N+K]}$ be the exact and noisy composite kernels for the full deterministic scan via MwG samplers under Assumptions \ref{assumption_bounded_support} and \ref{assumption_edge_probability}.
 These satisfy:
\begin{equation}\label{eq:thm_scan_2}
 \|P_{[N+K]} - \tilde{P}_{[N+K]}\| \leq \kappa b N^2 \ ;
\end{equation}
for a suitable positive finite constant $\kappa$ which does not depend on either $b$ or $N$.
\end{corollary}
\begin{proof}
 This is proved using Proposition \ref{prop_scan} and Theorem \ref{theorem_acceptance_probs_1}.
 There are $N$ latent position updates, whereby each of the kernels is upper bounded by $\kappa' b N$, for a suitable positive constant $\kappa'$.
 This gives a composite kernel with an upper bound of $\kappa' b N^2$.
 In addition, there are $K<\infty$ global parameters, whose number does not depend on $N$ or $b$, each yielding an upper bound of $\kappa'' b N^2$, for a suitable positive constant $\kappa''$.
 So, the upper bound for all combined updates is still $\kappa b N^2$, where $\kappa = \max\left\{ \kappa', K\kappa''\right\} < \infty$.
\end{proof}

Now we can state our main result for \texttt{NoisyLPM}, which connects Corollary \ref{corollary_2} with the theory on noisy MCMC.
\begin{theorem}
\label{thm:main}
Let $P$ be the exact MwG composite kernel which operates on $\mathcal{S}=\Scal_{\psi}\times \Scal_{\mathcal{Z}}$ and $\tP$ be the corresponding kernel of \texttt{NoisyLPM}. If the LPM satisfies Assumptions \ref{assumption_bounded_support} and \ref{assumption_edge_probability}, and if $P$ is uniformly ergodic (Assumption \ref{assumption_uniform}), then for any starting point $\btheta\in\Scal_{\psi}\times \Scal_{\mathcal{Z}}$ and any $t>0$:
\begin{equation}
\label{eq:main_thm_lpm}
 \|\delta_{\btheta} P^t - \delta_{\btheta} \tilde{P}^t\|
  \leq \left( \lambda + \frac{C\tau^\lambda}{1-\tau}\right) \kappa b N^2
\end{equation}
where $\lambda = \lceil{\log\left( 1/C \right)}\slash{\log(\tau)}\rceil$ depends on the exact sampler convergence properties, and $\kappa$ is a positive constant that does not depend on either $b$ or $N$.
\end{theorem}
\begin{proof}
Assumptions \ref{assumption_bounded_support} and \ref{assumption_edge_probability} guarantee that Corollary \ref{corollary_2} holds. Then, uniform ergodicity confirms that an LPM version of Corollary \ref{corollary_mitrophanov} exists, hence Eq. \eqref{eq:main_thm_lpm} holds true.
\end{proof}

\begin{remark}
 The significance of Theorem \ref{thm:main} is two-fold.
 On the one hand, for a fixed $N$, the error upper bound can clearly be made arbitrarily close to zero by reducing the grid parameter $b$, denoting the sidelength of the boxes.
 This confirms that using a finer grid reduces the approximation error.
 On the other hand, since the constant $\kappa$ does not depend on $N$, this result emphasises that, asymptotically, the error grows at most with the squared number of nodes. This suggests that, as a worst case scenario, the error can be kept constant (as $N$ increases) by keeping $b$ on the order of $1/N^2$, or, equivalently, a number of boxes\footnote{Note that, for a fixed $b$, we create $M^2$ boxes by partitioning each axis into $M$ segments of length $b$. The number of segments $M$ satisfies: $M = 2S/b \approx 2SN^2$, and so $M^2 \sim \mathcal{O}(N^4)$.} in the order of $N^4$.
 However, we show in the simulations (in particular in Section \ref{sec:gibbs_samplers_small}) that this bound is very conservative.
\end{remark}

\subsection{Note on the uniform convergence assumption}
Assumption \ref{assumption_uniform} is usually strong in the context of MCMC algorithms.
However, since the state space is compact (see Assumption \ref{assumption_bounded_support}),
it is easy to show that the convergence of the Gibbs kernel $P_{[R]}$ to $\pi$ is uniform.
Even though this result is not surprising, we could not identify a specific entry in the literature providing a rigorous proof of this fact.
For completeness, we include Theorem \ref{thm_uniform} in Appendix \ref{app:thm_uniform}.

\section{Experiments}\label{sec:experiments}
In this section we propose three simulation studies to characterise the bias introduced by our approximation, and to gauge the gain in computing time achieved.
We consider an LPM characterised by two global parameters $\boldsymbol{\psi} = \left( \beta,\theta \right)$ which determine the edge probabilities as follows:
\begin{equation}\label{eq:edge_prob_2}
 \log\left( \frac{p\left( \textbf{z}_i, \textbf{z}_j; \beta, \theta \right)}{1-p\left( \textbf{z}_i, \textbf{z}_j; \beta, \theta \right)} \right) := \beta - e^{\theta}d\left( \textbf{z}_i, \textbf{z}_j \right).
\end{equation}
Here, $\beta\in \mathbb{R}$, $\theta\in \mathbb{R}$, and $d$ denotes the Euclidean distance between the two latent positions.

A priori, the latent positions are IID variables distributed according to a truncated Gaussian, as shown in \eqref{eq:prior_z_1}.
The proposal distributions used in the sampling procedures are also truncated Gaussians, defined as random walks over the parameter space, but we note that other proposals may be considered.
We fix both the threshold parameter $S$ and the standard deviation $\gamma$ to $1$.
This choice does not hinder the flexibility of the model; in fact, the likelihood parameter $\theta$ directly regulates the magnitude of the effect of the latent space.
In other words, $e^{\theta}$ may simply be considered as the standard deviation for the latent positions.
The likelihood parameters $\beta$ and $\theta$ are assumed to be independent a priori, and both distributed according to non-informative Gaussian priors with fixed large standard deviations.

We note that the model specification considered does not completely satisfy Assumption \ref{assumption_bounded_support}, since, for example, the supports of $\beta$ and $\theta$ are not bounded.
However, we argue that large values of these parameters correspond to degenerate LPMs, which are of little interest in practical situations, and highly unlikely to occur.
In other words, the extreme values of the LPM parameters do not play a role and do not affect the MCMC estimation unless the observed graph is degenerate or near-degenerate.

\subsection{Study 1: likelihood approximation}\label{sec:study_1}
In the first study, we focus only on the approximation of the log-likelihood, i.e. we analyse the error introduced when \eqref{eq:likelihood_1} is replaced with the noisy counterpart in \eqref{eq:noisy_likelihood_1}.

First, we generate random LPMs with global parameters set to $\beta = 0.5$ and $\theta = \log\left( 3 \right)$, and with latent positions drawn uniformly in the rectangle $\mathcal{S}_{\mathcal{Z}}$.
This combination of parameters yields realised networks where about $10\%$ of the possible edges appear.
We simulate $100$ networks for each value of $N$ varying in the set $\left\{ 100, 250, 500, 1000, 2500, 5000, 10000 \right\}$.
For each of these realised networks, we evaluate the exact log-likelihood function derived from \eqref{eq:likelihood_1} for the true parameter values.

On each axis, the interval $[-1,1]$ is segmented in $M=8$ adjacent intervals of the same length, hence obtaining a grid of $64$ squared boxes of side length $1/4$.
The noisy log-likelihood derived from \eqref{eq:noisy_likelihood_1} is thus evaluated using such grid.
In fact, the same procedure is repeated on the same networks using various grids determined by $M$ in the set $\left\{ 8, 16, 32, 64 \right\}$.
Note that the highest values of $N$ and $M$ are rather extreme: $N=10000$ gives networks so large that even storing or working with the adjacency matrix is computer intensive (in fact our implementation does not require calculating the adjacency matrix at any stage); $M = 64$ gives a grid which contains $4096$ boxes, so, for several $N$ values, we would have more boxes than data points. This scenario is proposed only to provide a more complete assessment, since having these many boxes (compared to nodes) defeats the whole purpose of applying our procedure.

A preliminary plot is provided in Figure \ref{fig:sim_like_3}.
\begin{figure}[!htb]
\centering
\includegraphics[width=0.495\textwidth]{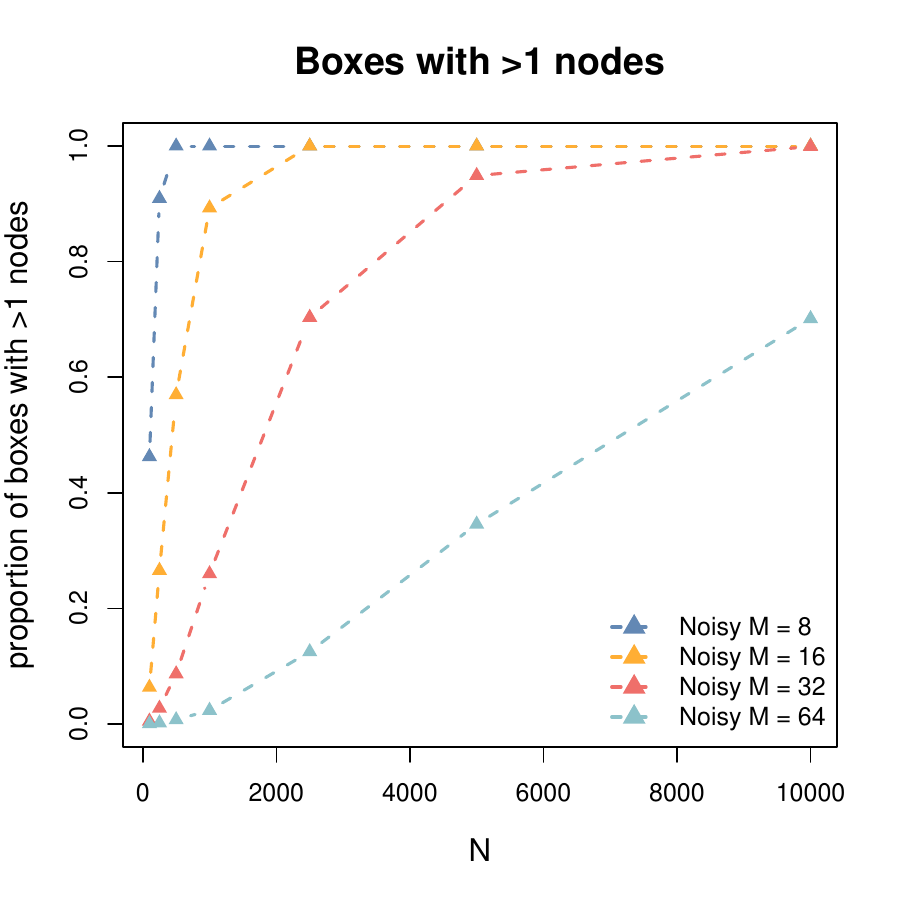}
\caption{\textbf{Simulation study 1}. Proportion of grid boxes that contain more than $1$ node.}
 \label{fig:sim_like_3}
\end{figure}
This plot shows the proportion of boxes that contain $2$ or more nodes.
This is derived without using the log-likelihood, but simply by constructing the grid over the randomly generated data.
The plot emphasises that, at least for the starting configuration of the noisy algorithms, most boxes would contain more than $1$ node.
The proportion seems to converge to $1$ rather quickly with $N$ increasing, with the only exception being given by the extreme value $M=64$.

A proportion close to $1$ is ideal for \texttt{NoisyLPM}, since only these boxes guarantee a gain in computational efficiency.
By contrast, boxes containing $0$ or $1$ nodes would require the algorithm to perform inefficient steps which are less convenient than the corresponding ones under a standard likelihood calculation.
The figure thus confirms that, in theory, good computational improvements should be obtained for very small $M$ values.

Figure \ref{fig:sim_like_2} shows the log-likelihood error and the average log-likelihood computing times for all of the combinations of $N$ and $M$.
\begin{figure}[!htb]
\centering
\includegraphics[width=0.495\textwidth,page=1]{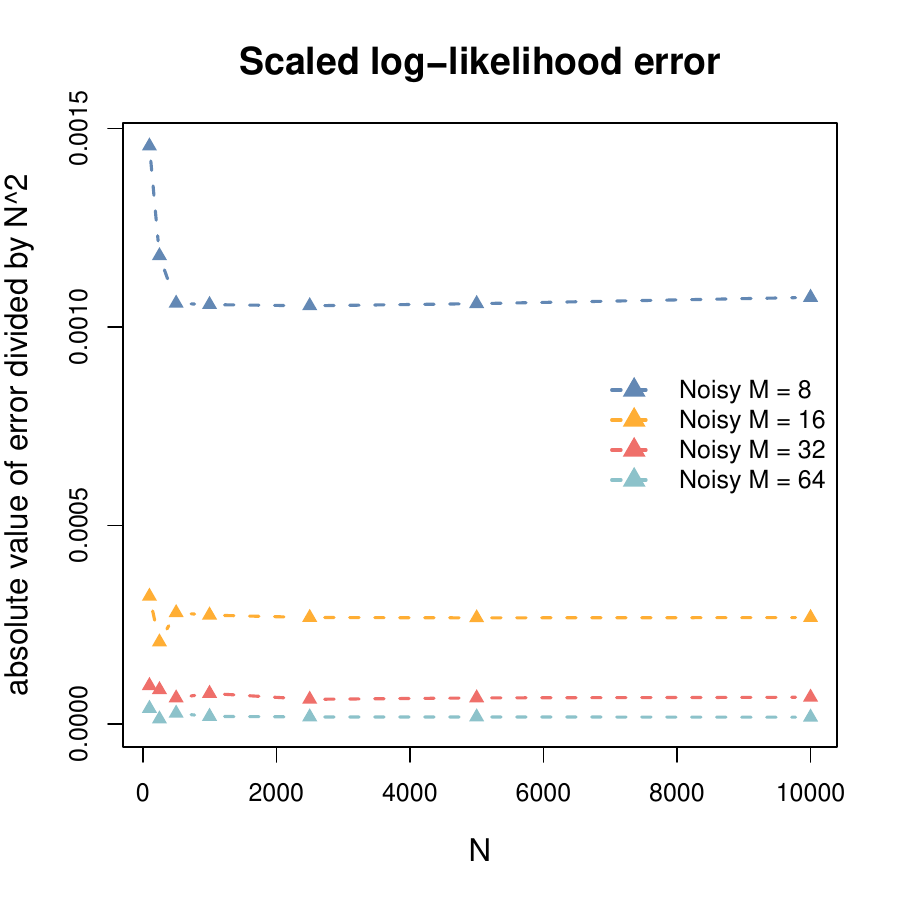}
\includegraphics[width=0.495\textwidth,page=1]{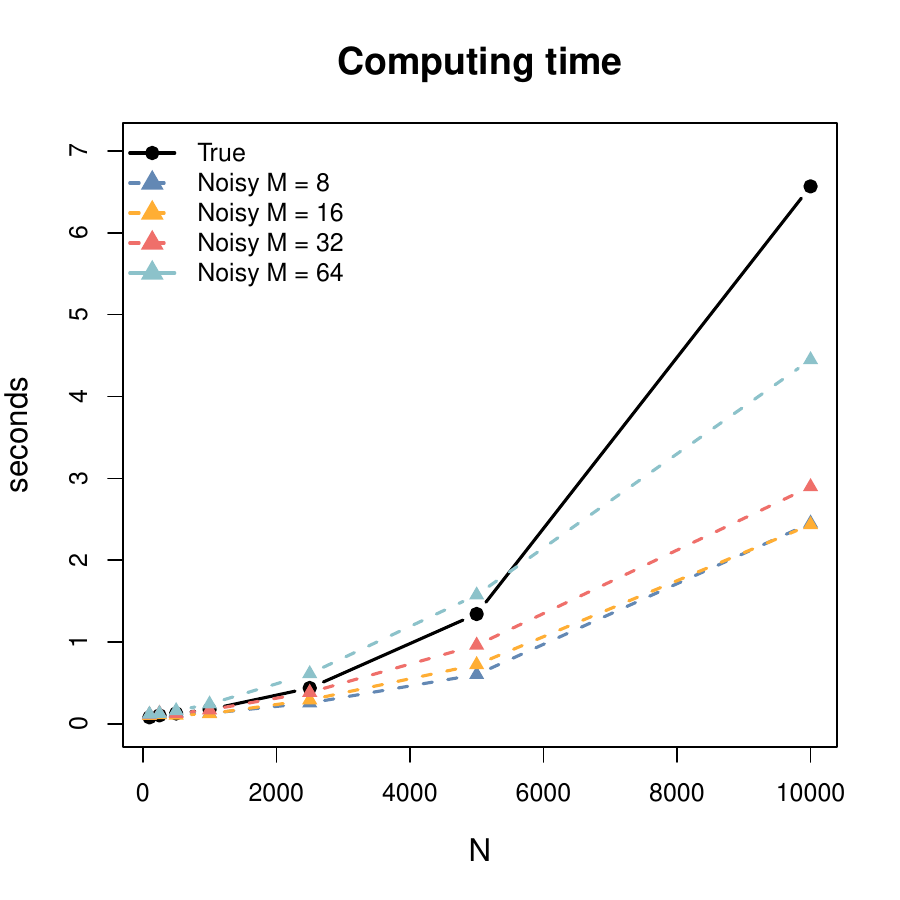}
\caption{\textbf{Simulation study 1}. The absolute value of the log-likelihood error, divided by $N^2$, is shown on the left panel. The right panel shows instead the average (across $100$ networks) computing time for the same log-likelihood evaluations.}
 \label{fig:sim_like_2}
\end{figure}
The left panel of this figure shows the absolute value of the error, divided by $N^2$.
The rescaling confirms very clearly that the error is asymptotically quadratic in the number of nodes.
The most accurate algorithm is obtained with $M=64$, which gives a very minimal average error on each likelihood term.

The right panel shows instead the computing time for one log-likelihood calculation, averaged out across the $100$ repetitions.
In this plot we can see that the computational complexity is highest for the standard algorithm (labelled as \texttt{True}): we know that the complexity of this particular algorithm is quadratic in the number of nodes.
The noisy algorithms all exhibits a computational complexity of a lower order, and an overall lower computing time.
These algorithms require the construction and maintainment of the grid structure, which involves some roughly constant computing time.
As a consequence, for $M=64$, the noisy algorithm gives a convenient trade-off only when $N$ becomes very large.
Again, this is reasonable since this particular value of $M$ is extreme and would not be considered in applications, unless $N$ is also especially large.
Figure \ref{fig:sim_like_2} confirms two basic but fundamental asymptotical facts: on the one hand, the log-likelihood errors can be arbitrarily reduced by choosing a finer grid, for a fixed $N$. On the other, for any grid fineness (fixed $M$), there exists an $N$ large enough such that the noisy algorithm is faster than the standard one.

\subsection{Study 2: Metropolis-within-Gibbs asymptotics}\label{sec:gibbs_samplers_small}
In the second simulation study, we aim at characterising the estimation errors using the complete MCMC sampling procedure.
This means that we run the MwG sampler on a number of networks, for both the non-noisy (which we take as ground truth) and noisy procedures.
Then, we compare noisy and exact posterior distributions using their means. In particular, we quantify the discrepancy by calculating the Mean Squared Error (MSE) from $100$ replications of both chains, for each parameter setting.

Since running the complete inferential procedure requires a much higher computational cost, we only provide these results for relatively smaller networks and coarser grids.
However, in our third simulation study, we provide further evidence that the results hold also for larger networks.

The setup of this study is as follows.
We consider networks of $N$ nodes where $N$ varies in the set $\left\{ 20, 40, 60, 80, 100\right\}$.
For each value $N$ we generate $100$ networks using $\beta = 2.5$ and $\theta = \log\left( 3.5 \right)$, and run the exact MwG sampler for $20{,}000$ iterations.
The first $10{,}000$ iterations are discarded as burn-in, and only one draw every $10$-th is stored to be kept in the final sample.
We consider grids defined by $M$ in the set $\left\{ 4, 8, 12\right\}$, and run the three corresponding samplers using the same number of iterations.
Although we do not check convergence diagnostics for each individual generated network and posterior sample, we argue that convergence is satisfactory across all experiments.
Besides, this aspect is not critical since we can study the errors even if convergence is not reached; or, in other words, we are studying the errors after $20{,}000$ MCMC iterations.

Figure \ref{fig:sim_mcmc_small} shows the results that we obtained.
\begin{figure}[!htb]
\centering
\includegraphics[width=0.65\textwidth,page=1]{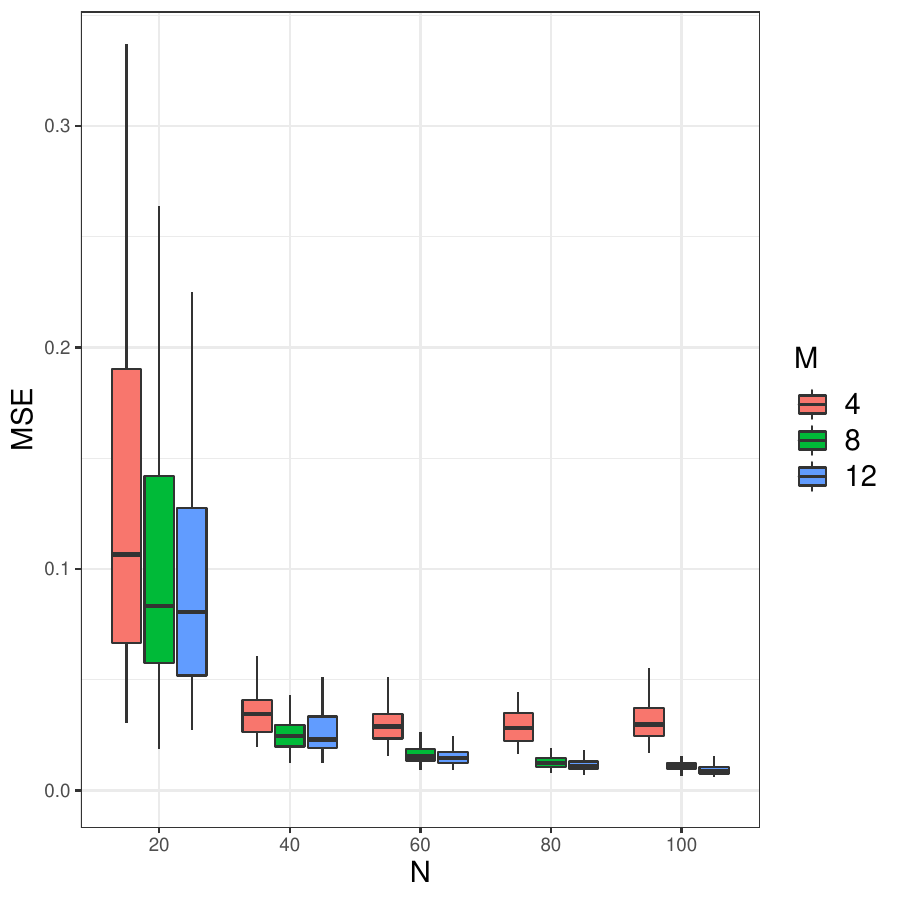}
\caption{\textbf{Simulation study 2}. The Mean Squared Error, where the mean is calculated across the $N$ nodes. The boxplot represents the variability across the $100$ repetitions, for each combination of $N$ and $M$.}
 \label{fig:sim_mcmc_small}
\end{figure}
Even though Theorem \ref{thm:main} suggests that $M$ should scale with $N^2$, the evidence from the simulations highlight a much more positive outcome, where, in fact, the errors do not increase with $N$ for $M = 8$ and $M = 12$. The approximation with $M=4$ ends up being too rough, as the posterior mean MSEs slightly increases between $N=80$ and $N=100$.
The bound of Theorem \ref{thm:main} offers a scaling which is, as expected, too conservative. Recall that Theorem \ref{thm:main} was essentially obtained by bounding the distance between both Markov kernels. Thus, since more data points propagate a larger discretisation error, the approximation offered by the noisy Markov kernel and the exact one appears to get rougher too.

One fundamental aspect that is not taken into account by our analysis is that more data points also bring more information available for inference. Our reasonably regular models imply that both noisy and exact posterior distributions see their probability mass concentrate towards the Maximum Likelihood Estimator (MLE) as $N$ increases, in an asymptotic regime typically described by a Bernstein-von Mises concentration result. The fact that $M$ needs to increase with $N$ is thus necessary since, both MLEs being different, the total variation distance between both posterior distributions will eventually (as $N$ gets larger) increase. What this study shows is that before entering this asymptotic regime, both posteriors actually get first closer. In fact, the discretisation approximation is clearly overcompensated by that concentration phenomenon, since keeping $M$ constant one still observes an improvement as $N$ increases. We suspect that for $M=4$, the Bernstein-von Mises regime kicks in after $N=80$. We also suspect that for a large enough $N$, a similar observation could be done for $M=8$ and $M=12$. 

For the noisy posterior to keep track with the exact one even in the asymptotic regime, it could be possible that a (much) less aggressive scaling than $M=\mathcal{O}(N^2)$ is sufficient. One way to answer that question could be to search for tighter bounds between both posterior distributions using their Bernstein-von Mises approximation, hence by totally bypassing the analysis of both Markov kernels.

A second aspect that transpires from this simulation study is a confirmation that finer grids will give more accurate results, with major improvements happening with relatively coarse grids. We expand more on this aspect in the next simulation study, where we show that we can choose very small values of $M$ (relative to $N$), and still obtain very accurate inference results, with only a fraction of the computing time.

\subsection{Study 3: Metropolis-within-Gibbs error}\label{sec:gibbs_samplers}
As data, we use three artificial networks, which are generated so that each of them has network density close to $10\%$.
In this simulation study we do not use any repetitions, so we analyse exactly three networks.
Another difference with the previous setup is that, in this study, node $1$ is assumed to be located exactly at the origin of the space, for comparison purposes.
The number of nodes $N$ of the networks is set to $200$, $400$ and $600$, respectively.
For the \texttt{NoisyLPM}, we consider three different grid structures: the number of intervals $M$ in each axis varies in the set $\left\{ 8,12,16\right\}$ (the results are shown for $M = 8$ and $M = 16$, but all results are available from the authors upon request).

The non-noisy MwG sampler and the \texttt{NoisyLPM} are run on each dataset for a total of $200{,}000$ iterations.
The first $100{,}000$ iterations are discarded as burn-in, and only one draw every $10$-th is stored to be kept in the final sample.
Eventually, all of the algorithms are bound to return a collection of $10{,}000$ draws for each model parameter.
In this simulation study, we checked and confirmed MCMC convergence individually for each of the posterior samples.

Figure \ref{fig:sim_mcmc_node} shows the posterior densities for the node located in the centre of the space.
\begin{figure}[!htb]
\centering
\includegraphics[width=0.49\textwidth]{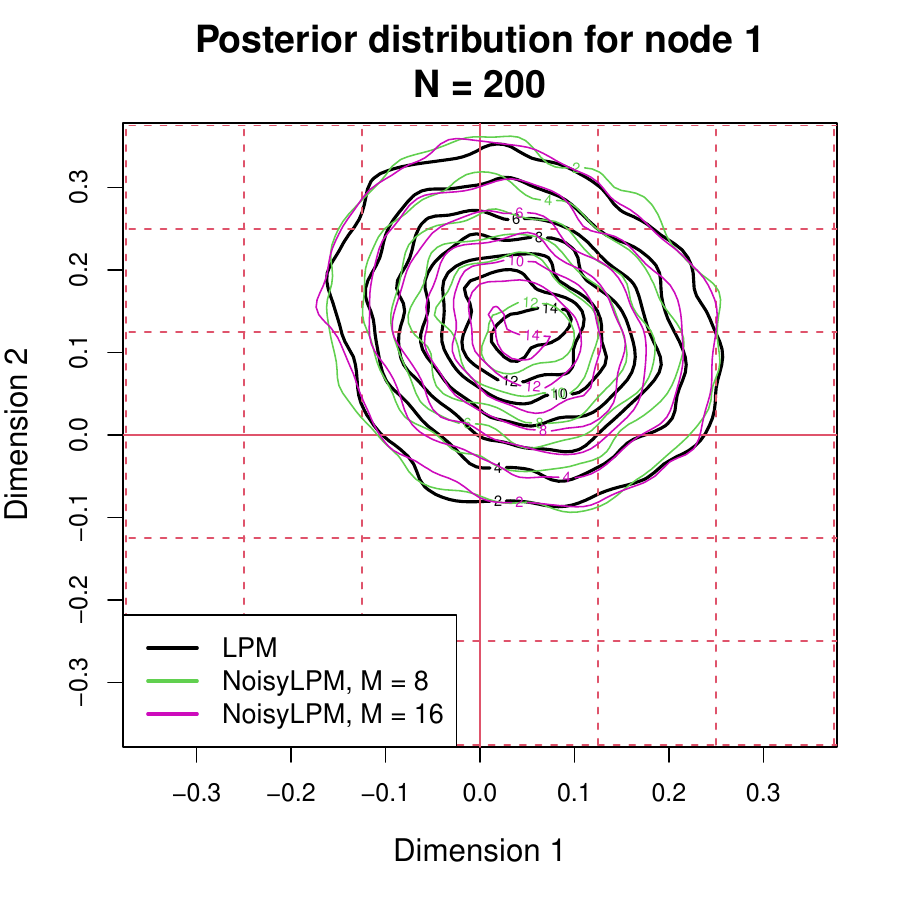}
\includegraphics[width=0.49\textwidth]{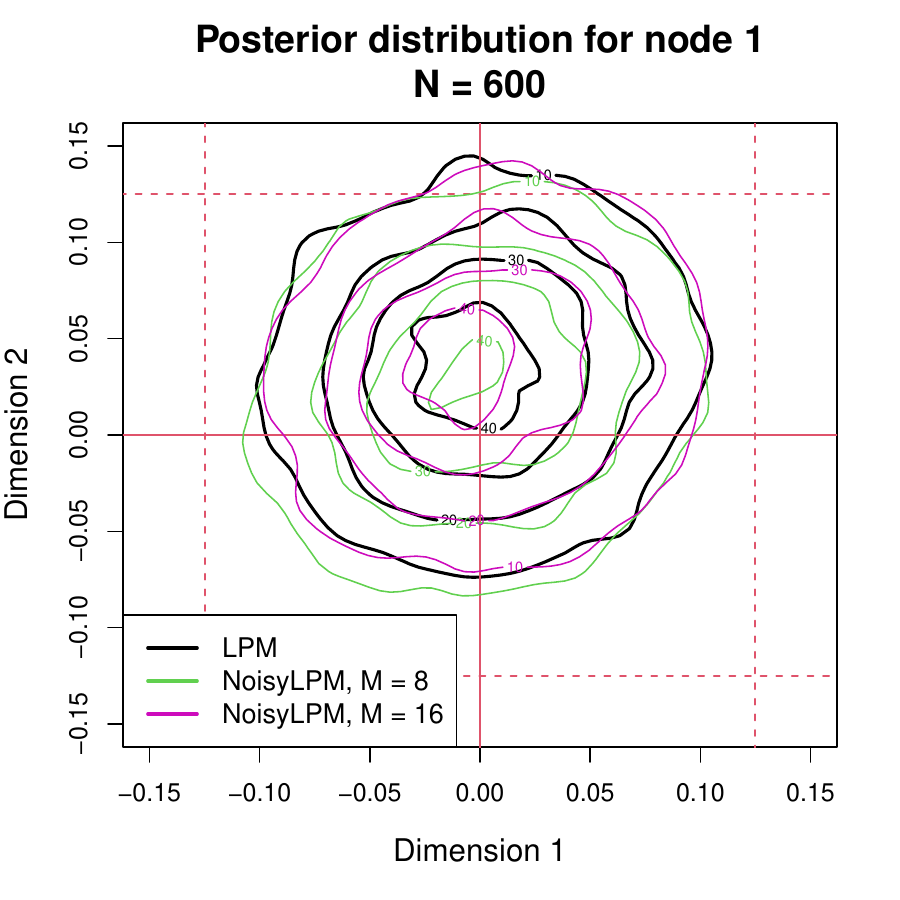}
\caption{\textbf{Simulation study 3}. Posterior densities for the node in the centre. }
 \label{fig:sim_mcmc_node}
\end{figure}
The two \texttt{NoisyLPM} posterior densities shown are extremely similar to the ground truth, proving that the uncertainty in the positioning is not necessarily amplified by the approximation.

Figure \ref{fig:sim_mcmc_positions} focuses instead the (posterior) average position as a point estimator, and compares the estimated positions of all nodes in the ground truth and noisy case.
\begin{figure}[!htb]
\centering
\includegraphics[width=0.49\textwidth,page=1]{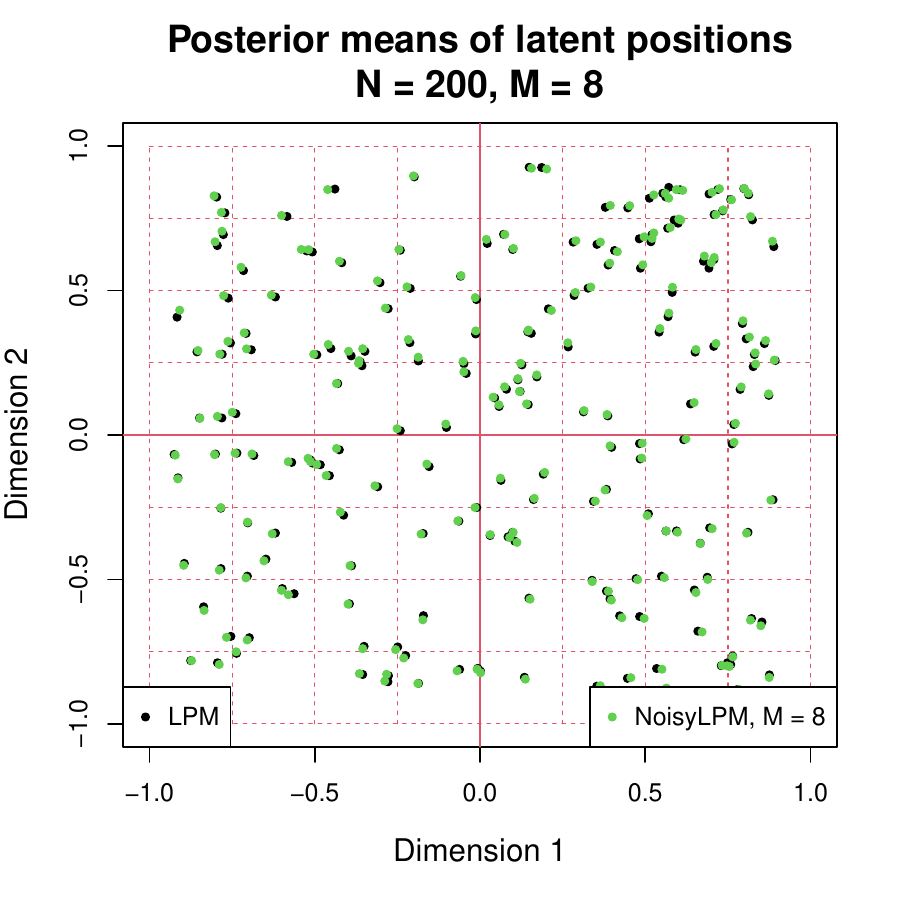}
\includegraphics[width=0.49\textwidth,page=3]{sim_N_200_positions_true_to_average.pdf}\\
\includegraphics[width=0.49\textwidth,page=1]{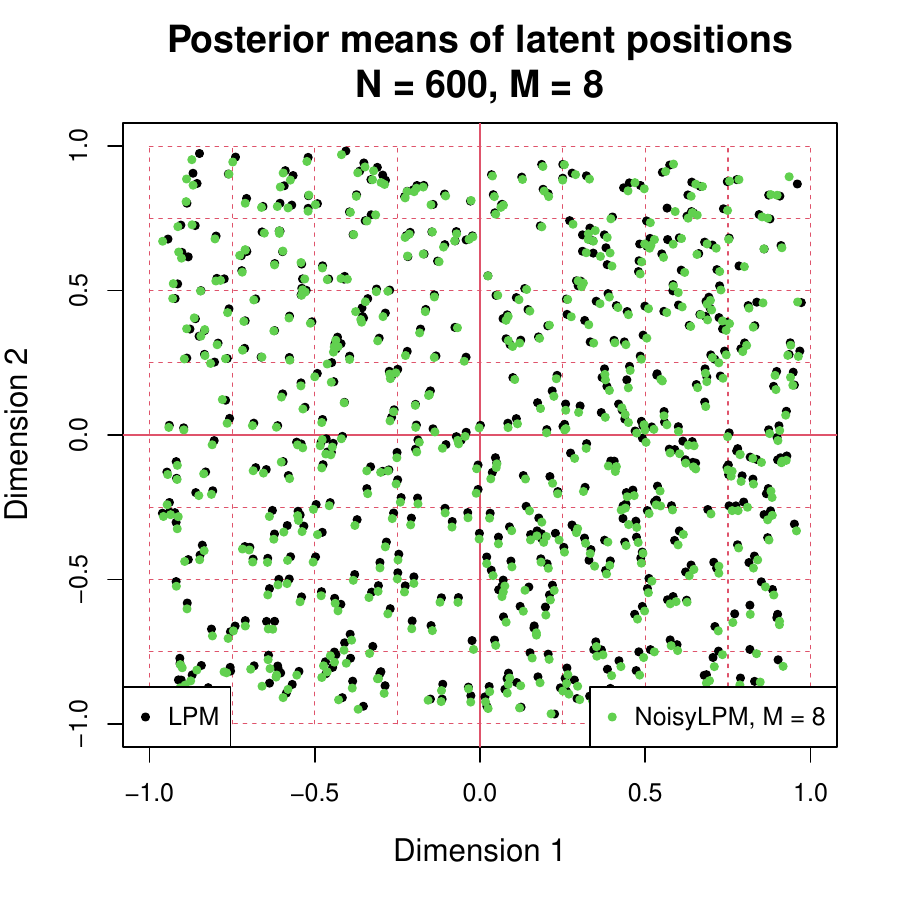}
\includegraphics[width=0.49\textwidth,page=3]{sim_N_600_positions_true_to_average.pdf}\\
\caption{\textbf{Simulation study 3}. Comparison between ground truth and noisy estimates of the positions.
The black circles correspond to the posterior means of the positions in the ground truth configuration, whereas the green and pink nodes correspond to the noisy counterparts.}
 \label{fig:sim_mcmc_positions}
\end{figure}
Again, the approximation appears to have very limited consequences on the correctness of the results.
In particular, the estimation error is almost non-existent when $M=16$.

Figures \ref{fig:sim_mcmc_beta} and \ref{fig:sim_mcmc_theta} illustrate the posterior densities for the global parameters $\beta$ and $\theta$.
\begin{figure}[!htb]
\centering
\includegraphics[width=0.49\textwidth]{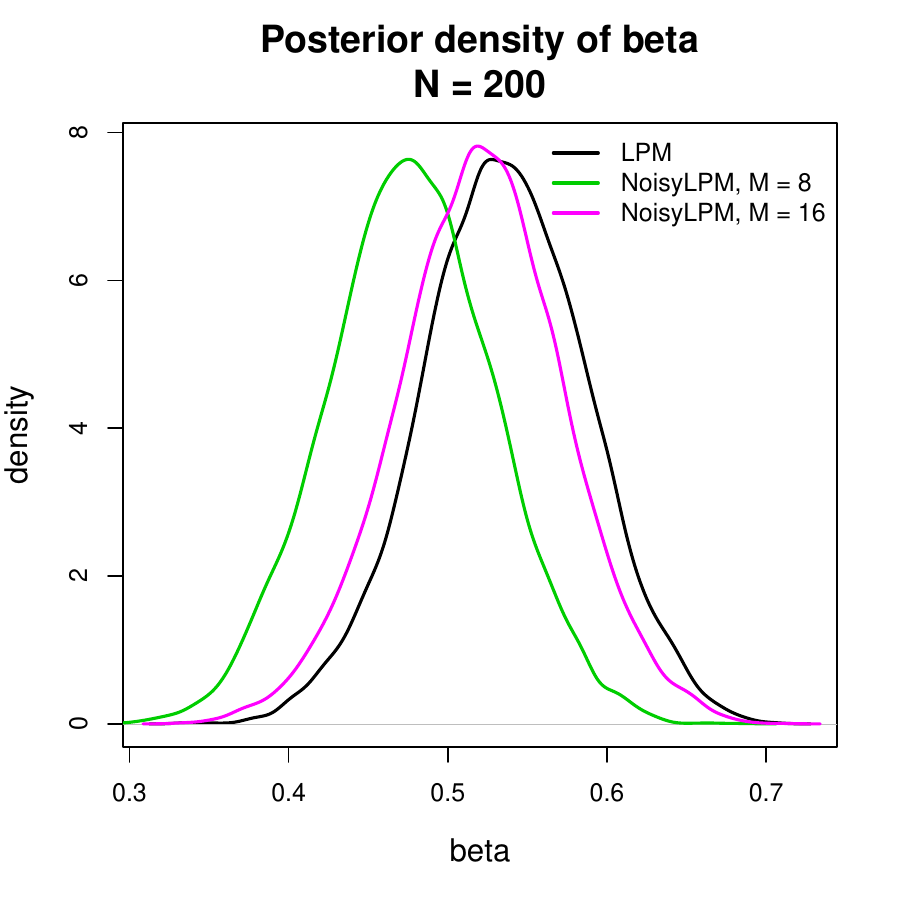}
\includegraphics[width=0.49\textwidth]{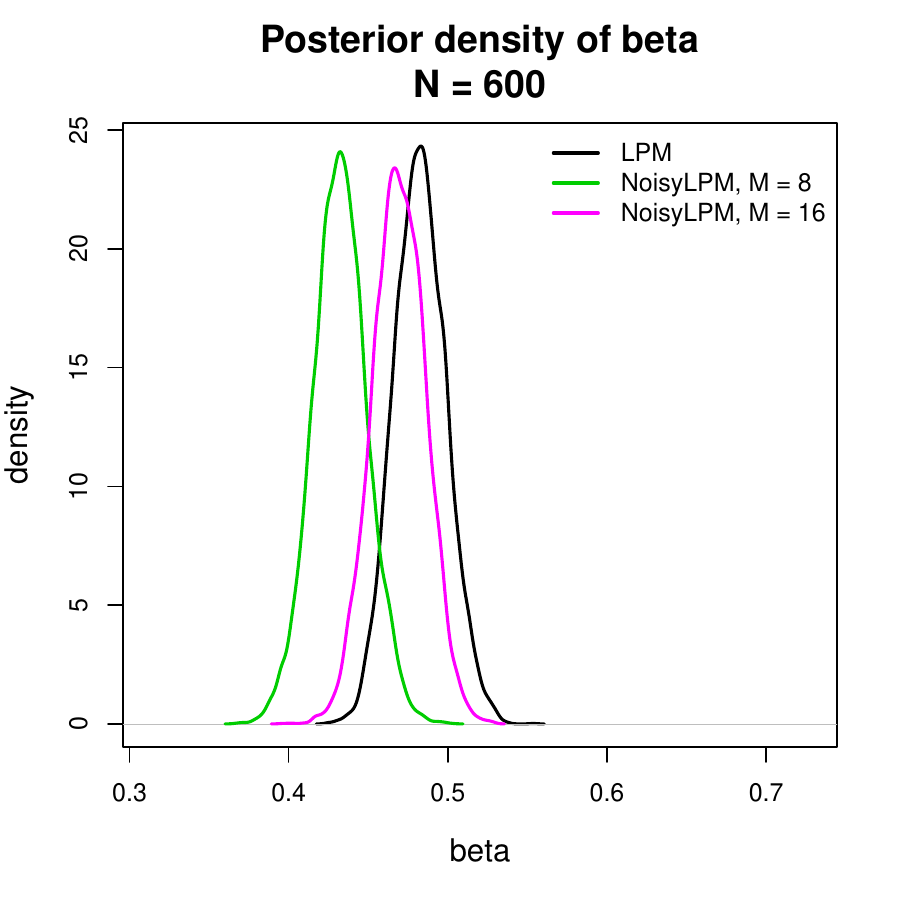}
\caption{\textbf{Simulation study 3}. Posterior densities for $\beta$. Note the different scaling in the horizontal axis.}
 \label{fig:sim_mcmc_beta}
\end{figure}
\begin{figure}[!htb]
\centering
\includegraphics[width=0.49\textwidth]{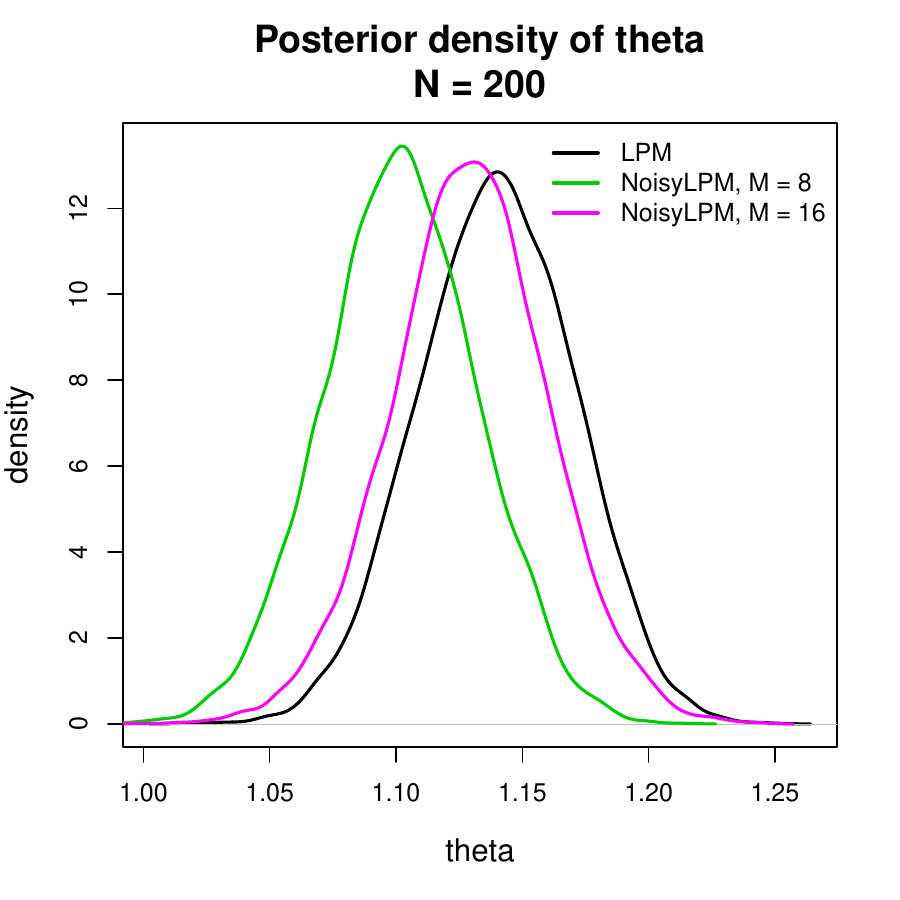}
\includegraphics[width=0.49\textwidth]{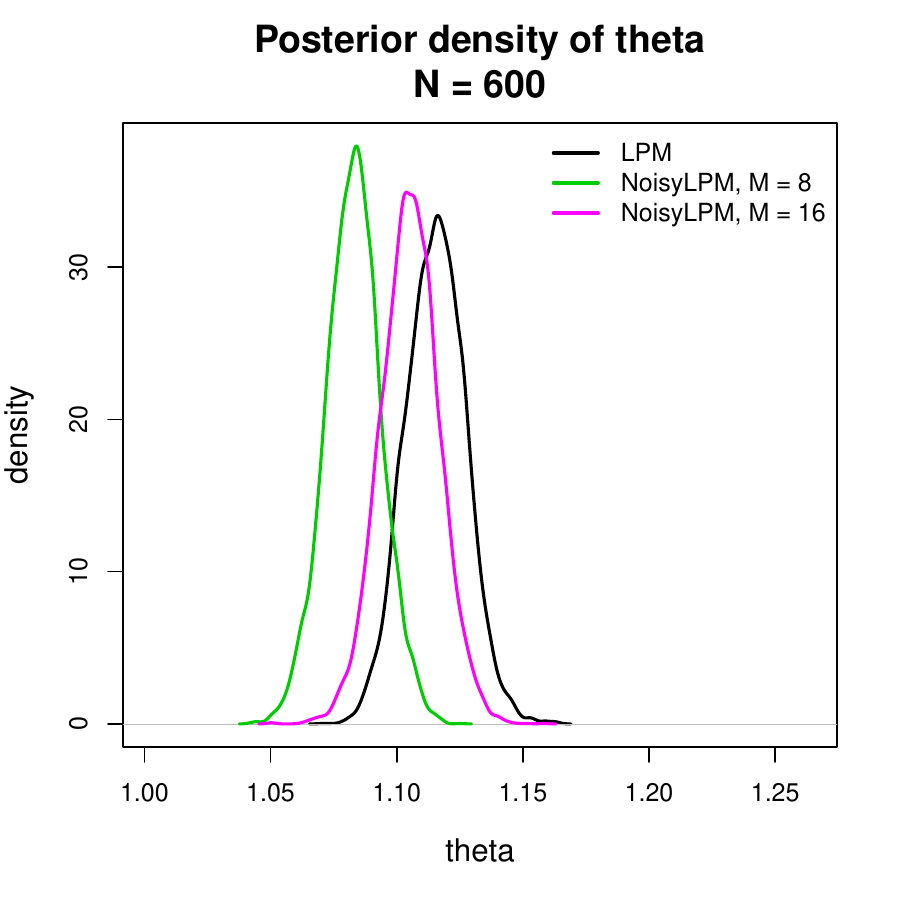}
\caption{\textbf{Simulation study 3}. Posterior densities for $\theta$. Note the different scaling in the horizontal axis.}
 \label{fig:sim_mcmc_theta}
\end{figure}
Note that, in both figures, the horizontal axes of the plots are on different scales.
In fact, these plots confirm that the uncertainty on global parameters tends to vanish as $N$ increases, for both non-noisy and noisy algorithms.
As expected, a larger $M$ gives results closer to the ground truth.

We further analyse the results by comparing the estimated edge probabilities in Figure \ref{fig:sim_mcmc_edges}.
\begin{figure}[!htb]
\centering
\includegraphics[width=0.425\textwidth]{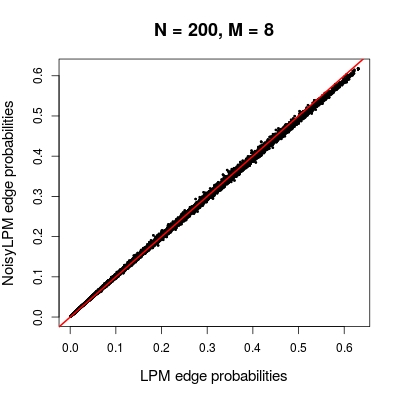}
\includegraphics[width=0.425\textwidth]{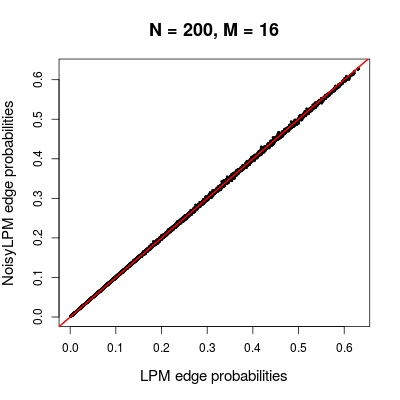}\\
\includegraphics[width=0.425\textwidth]{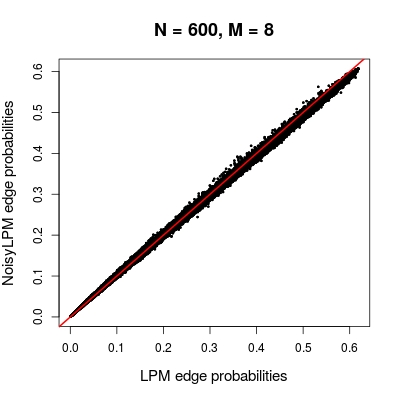}
\includegraphics[width=0.425\textwidth]{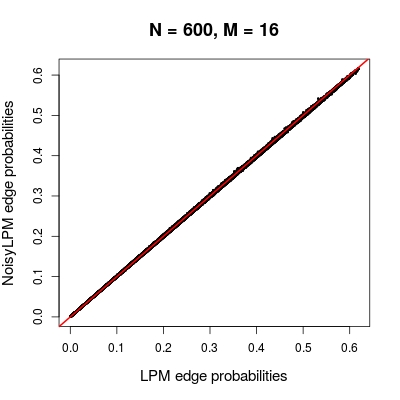}\\
\caption{\textbf{Simulation study 3}. Comparison between ground truth and noisy estimates of the edge probabilities.
These estimates are obtained by pluggin-in the posterior mean estimates of the model parameters in \eqref{eq:edge_prob_2}.}
 \label{fig:sim_mcmc_edges}
\end{figure}
These plots also confirm the correctness of the noisy procedure, and the limited effects of the approximation on the results.

Finally, in Table \ref{tab:sim_mcmc_table_1} we show the computing time required for each sampler.
The highest gain is achieved for $M=8$ and $N=600$, where the \texttt{NoisyLPM} is roughly three times faster than the benchmark.
As we will show in the next section, the gain can become substantial when larger networks are considered.
\begin{table}[htb]
\centering
\begin{tabular}{ccccc}
  \specialrule{.1em}{0em}{0em}
          $N$ & Ground truth & \multicolumn{3}{c}{\texttt{NoisyLPM}}   \\
           &  & $M = 8$ & $M = 12$ & $M = 16$  \\
  \specialrule{.1em}{0em}{0em}
          200 & 2{,}310 & 1{,}669 & 2{,}252 & 2{,}767  \\
          400 & 7{,}242 & 3{,}515 & 5{,}458 & 7{,}673  \\
          600 & 14{,}347 & 4{,}718 & 7{,}501 & 11{,}825  \\
   \specialrule{.1em}{0em}{0em}
\end{tabular}
\caption{\textbf{Simulation study 3}. Seconds (rounded value) required to obtain $200{,}000$ iterations from each of the networks, for both algorithms.}
\label{tab:sim_mcmc_table_1}
\normalsize
\end{table}

\section{Coauthorship in astrophysics}\label{sec:coauthorship}
\subsection{Binary case}
The coauthorship network studied in this section was first analysed by \textcite{leskovec2007graph}.
The nodes correspond to authors, whereas the presence of an edge between two nodes means that the two researchers appear as coauthors on a paper submitted to arXiv, in the astrophysics category.
The network is by construction undirected and without self-edges.
The number of nodes is $18{,}872$, whereas the number of edges is $198{,}110$, corresponding to an average degree of about $21$.

We fit the LPM of Section \ref{sec:experiments} to this data using the \texttt{NoisyLPM} with $M = 16$.
First, we let the algorithm run for a large number of iterations.
We use this phase as burn-in, and to tune the proposal variances individually for each parameter until the corresponding acceptance probability lies between $20\%$ and $50\%$.
Then, we run the \texttt{NoisyLPM} for $50{,}000$ iterations, storing only one draw every $10$-th.
Trace plots and other standard convergence diagnostics suggest good mixing and good convergence of the chain to its stationary distribution.
In summary, for each latent position and global parameter, we obtain $5{,}000$ random draws that can be used to characterise the distribution of interest.

Figure \ref{fig:astro_z} shows the average latent positions for all of the nodes in the network.
\begin{figure}[!htb]
\centering
\includegraphics[width=\textwidth]{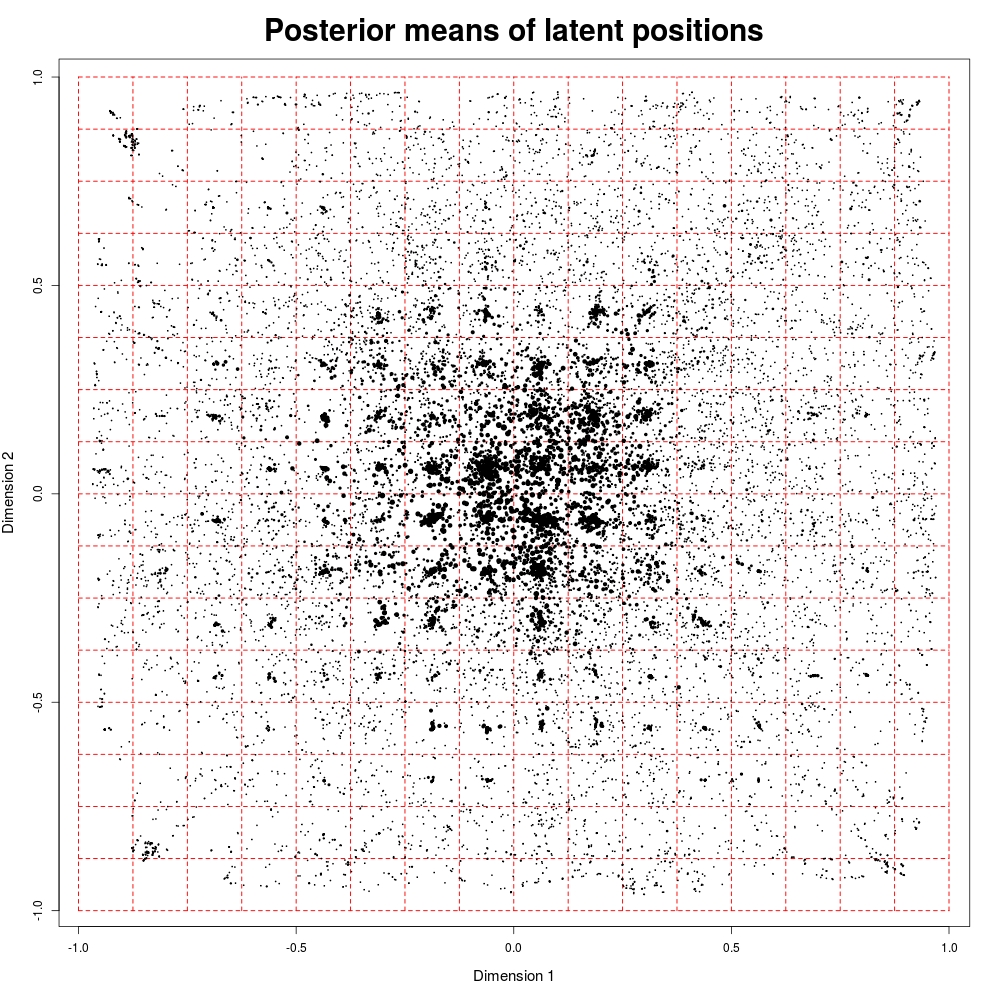}
\caption{\textbf{Astrophysics}. Average latent positions of the nodes with circle size proportional to the node's degree. The grid in dashed red line corresponds to the partitioning imposed.}
 \label{fig:astro_z}
\end{figure}
We point out that the nodes have a tendency to be distributed close to the centre of each box.
Quite reasonably, this is a natural consequence of our construction, since the centre of the boxes is used as a proxy to calculate the latent distances.
For example, if a node with a low degree is connected only to nodes allocated to the same box, it will tend to move towards the centre of the same box, since that would maximise the likelihood of those edges appearing.
More generally, we argue that, while the overall macro-structure of the latent space (i.e. the association of nodes to boxes, or the association of nodes to subregions of the space) is properly recovered,
the micro-structure, given by the relative positions of the nodes within each box, may not necessarily be accurate.

Figure \ref{fig:astro_beta_theta} shows instead the posterior densities for the global parameters $\beta$ and $\theta$.
\begin{figure}[!htb]
\centering
\includegraphics[width=0.45\textwidth]{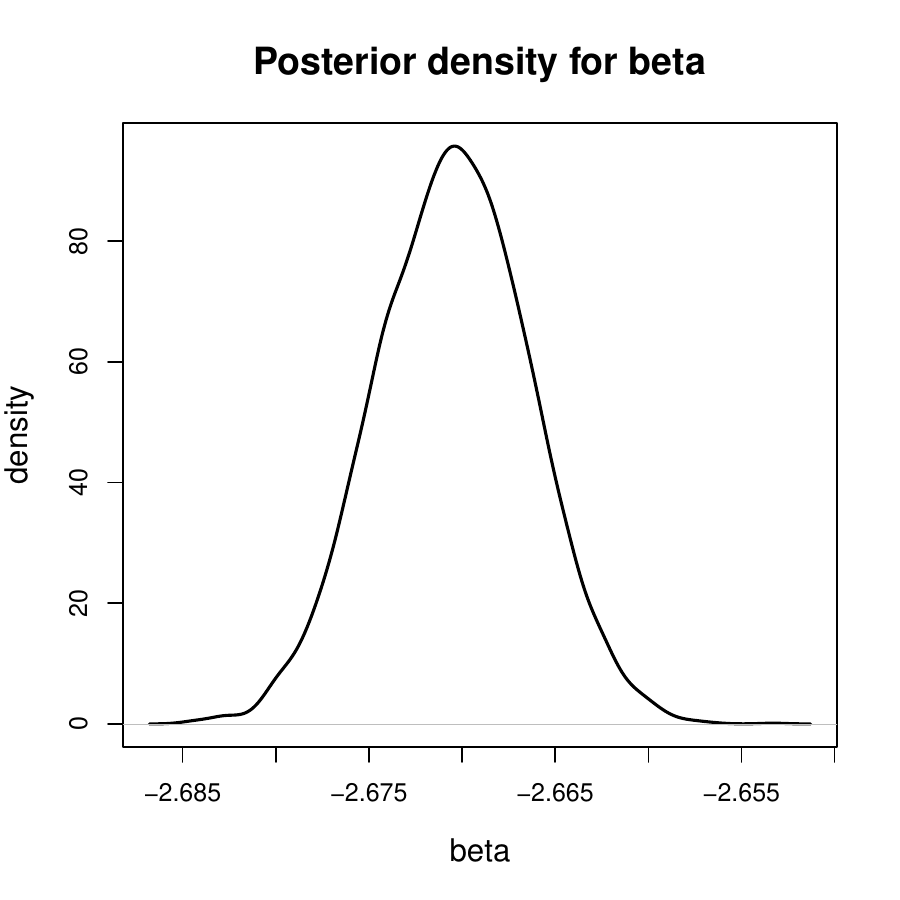}
\includegraphics[width=0.45\textwidth]{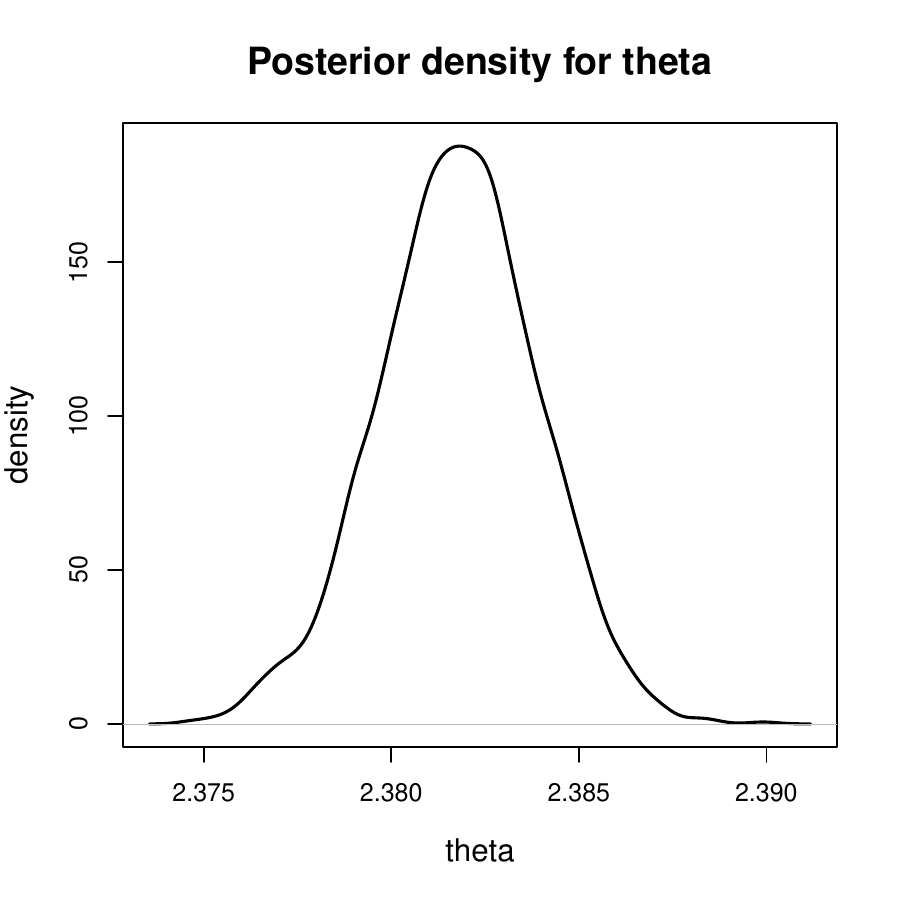}
\caption{\textbf{Astrophysics}. Posterior densities for the global parameters $\beta$ and $\theta$.}
 \label{fig:astro_beta_theta}
\end{figure}
We find the parameter $\theta$ to be rather large, signalling that the heterogeneity of the graph is well captured by the latent space.

The computing time required to obtain the sample was about $46$ hours ($3.3$ seconds per iteration).
After convergence of the Markov chain, we also ran the non-noisy MwG sampler for $50$ additional iterations, to compare the computational efficiency of the two methods.
The non-noisy MwG sampler required an average of $453$ seconds per iteration, corresponding to a theoretical $262$ days of computations for the full sample.
The vast difference between the two computing times highlights the scalability of our method, which extends the applicability of LPMs to networks of much larger sizes.

\subsection{Poisson case}
The coauthorship network data was in fact collected as a weighted graph, where the non-negative weights of edges represent the strength of the collaboration relations \footnote{Note that the edge weights are not exactly equal to the number of coauthored papers, however, they are constructed from this information through some rescaling. Details can be found in the paper by \textcite{newman2001scientific}. In our paper, we use these weights in a discretised way to fit our Poisson model.}.
This dataset offers then a nice setup to extend our methodology to a discrete Poisson setup, where edge weights are non-negative integers. The extension of our framework turns out to be rather straightforward, since the model can be defined as:
\begin{equation}\label{eq:poisson_case_1}
\begin{split}
 Y_{ij} &\sim Pois(\lambda_{ij}) \\
 \log \lambda_{ij} &= \beta - e^{\theta} d\left( \textbf{z}_{i}, \textbf{z}_{j} \right)
\end{split}
\end{equation}

As regards \texttt{NoisyLPM}, the implementation of the algorithm requires some modifications to take into account the fact that the grid does not partition the nodes into homogeneous boxes: the boxes now include different nodes which give different likelihood contributions based on their Poisson weight. To overcome this, we add an new dimension to our grid, whereby we split up each box into subgroups of nodes that carry the same Poisson weight. In other words, we further break down the partition to recover the homogeneity within each of the partition sets.

As a demonstration, we fit the model to the same data, this time using a grid obtained through $M = 8$: the latent space is shown in Figure \ref{fig:astro_w_z}.
\begin{figure}[!htb]
\centering
\includegraphics[width=\textwidth]{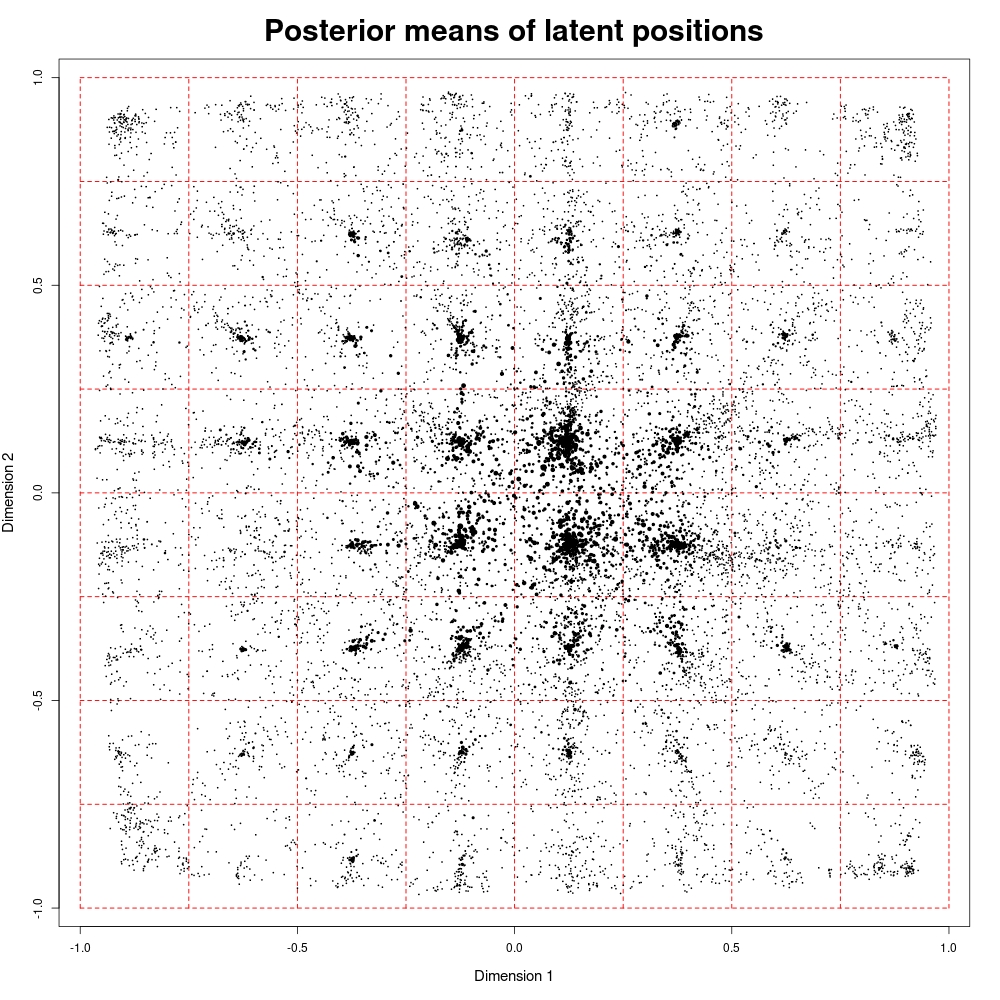}
\caption{\textbf{Astrophysics, weighted}. Average latent positions of the nodes with circle size proportional to the node's weighted degree. The grid in dashed red line corresponds to the partitioning imposed.}
 \label{fig:astro_w_z}
\end{figure}
The latent space for the weighted case appears to be equivalent to that obtained in the binary case, with a stronger clustering effect around the centres of the boxes, as a consequence of using a coarser grid. The posterior average of $\beta$ and $\theta$ are $-2.83$ and $2.34$, respectively.
This again signals a relatively strong latent space effect in the model, and thus a good fit to the data.

\section{Conclusions}
In this paper, we have introduced a new methodology to perform inference on latent position models.
Our approach specifically addresses a crucial issue: the scalability of the method with respect to the size of the network.
By taking advantage of a discretisation of the latent space, our proposed approach is characterised by a reduced computational complexity compared to the state of the art procedures.

The framework introduced heavily relies on several important results introduced in the context of noisy MCMC.
We have followed the core ideas of such strand of literature, and adapted the main results to the latent position model context, thereby giving theoretical guarantees for our proposed approximate method.
In particular, our results underline the existence of a trade-off between the speed and the bias of the noisy algorithm, whereby the user can arbitrarily increase the accuracy at the expense of speed.

Additionally, we have proposed applications to both simulated and real datasets.
When compared to the non-noisy algorithm, the noisy results did not show any relevant qualitative difference, yet they were obtained with a substantially smaller computing time.
While we describe and illustrate our method in the case of binary networks, we also provide an application to Poisson weighted networks, showing that the method could be extended to other common network structures.
The code for \texttt{NoisyLPM} is available from the public GitHub repository \textcite{noisylpm_github}.

A limitation of our work is that it might not cope well with an increasing number of latent dimensions.
This would imply an increase in the dimensionality of the grid, and, in turn, a much larger number of boxes.
Similarly, introducing nodal random effects or covariates would create additional heterogeneity within any given box, thus requiring further dimensions in the latent grid.
This would make our approach impractical.
However, we also point out that usually covariates can change the statistical analysis by diminishing the interpretability of the latent positions, since they may be critical features of the data.

Finally, our work can be easily extended to include different distributions on the latent space, such as Gaussian mixture models \parencite{handcock2007model} or different types of edge probabilities, or, more generally, networks factor models, such as the projection models of \textcite{hoff2002latent}.

\section*{Software}
The method has been implemented in \texttt{C++}, and it uses parallel computing through the library \texttt{OpenMPI}.
All of the computations described in the paper have been performed on a $8$-cores ($2.2$ GHz) Debian machine.
The code for \texttt{NoisyLPM} is available from the public GitHub repository \textcite{noisylpm_github}.

\section*{Acknowledgements}
Part of this research has been carried out while R. R. was affiliated with the Institute of Statistics and Mathematics,
Vienna University of Economics and Business, Vienna, Austria; and funded through the Vienna Science and Technology Fund (WWTF) Project MA14-031. This research was also supported by the Insight Centre for Data Analytics through Science Foundation Ireland grant SFI/12/RC/2289.

\printbibliography

\footnotesize
\appendix
\section{Appendix}
Before proceeding with two lemmas that will be useful in the proof of Proposition 2, we start with some elements of context.
Consider a state-space $(\Scal,\Acal)$ and define $\mathcal{M}_1$ the set of probability measures on $(\Scal,\Acal)$. A Markov kernel $P$ operates (from the left) on $\mathcal{M}_1$ by $\mu\mapsto \mu P:=\int \mu(\rmd \btheta)P(\btheta,\,\cdot\,)\in \mathcal{M}_1$. Define the total variation of $P$ as the number $\|P\|:=\sup_{\mu\in\mathcal{M}_1}\|\mu P\|=\sup_{\btheta\in\Scal}\|P(\btheta,\,\cdot\,)\|$ where for any measure $\mu$ defined on $(\Scal,\Acal)$, the positive number $\|\mu\|:=\sup_{A\in\Acal}|\mu(A)|$ denotes its total variation.
\begin{lemma}\label{lem1}
The total variation norm of a Markov kernel is $1$.
\end{lemma}
\begin{proof}
On the one hand, we have
$$
\|P\| = \sup_{\btheta\in\Scal}\|P(\btheta,\, \cdot\,)\| = \sup_{\btheta\in\Scal} \sup_{A \in \Acal} |P( \btheta, A )| \leq 1
$$
but, since $\Scal \in \Acal$, $ 1=P( \btheta, \Scal )\leq \sup_{A \in \Acal} P( \btheta, A) $,  for all $\btheta\in\Scal$, and we thus also have that
$$
1 = \sup_{\btheta\in\Scal} P( \btheta, \Scal ) \leq \sup_{\btheta\in\Scal} \sup_{A \in \Acal} |P( \btheta, A)| = \|P\| \ .
$$
\end{proof}

\noindent We now consider a signed Markov kernel $K$ which is the difference between two Markov kernels $P_1$ and $P_2$, that is  $K=P_1-P_2$. Define by $\mathcal{M}_0$ the set of signed measures on $(\Scal,\Acal)$ such that for all $\mu\in\mathcal{M}_0$, $\mu(\Scal)=0$. It is easy to see that $K$ operates (from the left) on $\mathcal{M}_0$ by $\mu\mapsto \mu K:=\int \mu(\rmd \btheta)K(\btheta,\,\cdot\,)\in \mathcal{M}_0$. Hence, we define the total variation of such a signed Markov kernel $K$ by $\|K\|=\sup_{\mu\in\mathcal{M}_0}\|\mu K\|=\sup_{\mu\in\mathcal{M}_0}\sup_{A\in \Acal}|\mu K(A)|$.
\begin{lemma}\label{lem2}
The total variation norm for signed kernels is submultiplicative: for two signed kernels $K_1$ and $K_2$ wih $K_1(\btheta,\Scal)=K_2(\btheta,\Scal)=0$, we have $\|K_1K_2\| \leq \|K_1\|\|K_2\|$.
\end{lemma}
\begin{proof}
First, note that 
$$
\sup_{\nu\in\mathcal{M}_0}\|\nu K_1K_2\|=\max\left[\sup_{\nu\in\mathcal{M}_0\,:\,\|\nu K_1\|>0}\|\nu K_1K_2\|,\sup_{\nu\in\mathcal{M}_0\,:\,\|\nu K_1\|=0}\|\nu K_1K_2\|\right]=\sup_{\nu\in\mathcal{M}_0\,:\,\|\nu K_1\|>0}\|\nu K_1K_2\|\,,
$$
since for all $\nu\in\mathcal{M}_0$, $\|\nu K_1\|=0$ implies that $\nu K_1$ is the null measure, which implies that $\nu K_1 K_2$ is also the null measure and thus that $\|\nu K_1 K_2\|=0$.
We have that
\begin{multline*}
  \|K_1K_2\|=\sup_{\nu\in\mathcal{M}_0}\|\nu K_1K_2\|=\sup_{\nu\in\mathcal{M}_0\,:\,\|\nu K_1\|>0}\|\nu K_1K_2\|=\sup_{\nu\in\mathcal{M}_0\,:\,\|\nu K_1\|>0}\frac{\|\nu K_1K_2\|}{\|\nu K_1\|}\|\nu K_1\|\\
  \leq \sup_{\nu\in\mathcal{M}_0\,:\,\|\nu K_1\|>0}\frac{\|\nu K_1K_2\|}{\|\nu K_1\|}\sup_{\nu\in\mathcal{M}_0}\|\nu K_1\|=\sup_{\nu\in\mathcal{M}_0\,:\,\|\nu K_1\|>0}\frac{\|\nu K_1K_2\|}{\|\nu K_1\|}\| K_1\|\,.
\end{multline*}
But note that for all $\nu\in\mathcal{M}_0$ with $\|\nu K_1\|>0$, we have for all $A\in\Acal$
$$
\frac{\nu K_1K_2}{\|\nu K_1\|}(A)=\iint \frac{\nu(\rmd \btheta) K_1(\btheta,\rmd \btheta')}{\|\nu K_1\|}K_2(\btheta',A)=
\int \nu'(\rmd \btheta')K_2(\btheta',A)=\nu' K_2(A)\,,
$$
where $ \nu'=\int \nu(\rmd \btheta)\frac{K_1(\btheta,\,\cdot\,)}{\|\nu K_1\|}$. This implies that
$\frac{\|\nu K_1K_2\|}{\|\nu K_1\|}=\|\nu' K_2\|$. But since $\nu'\in\mathcal{M}_0$,
$$
\frac{\|\nu K_1K_2\|}{\|\nu K_2\|}\leq \sup_{\nu'\in\mathcal{M}_0}\|\nu' K_2\|=\|K_2\|\,,
$$
which does not depend on $\nu$. The proof is completed by taking the supremum over all $\nu\in\mathcal{M}_0$ in the last inequality.
\end{proof}

\subsection{Proof of Proposition \ref{prop_noisy_corollary}}\label{app:prop_noisy_corollary}
\begin{proof}
By definition: $\|P - \tP\| = \sup_{\btheta\in \Scal}\|P\left( \btheta, \cdot \right) - \tP\left( \btheta, \cdot \right)\|$. Now, $P\left( \btheta, \cdot \right)$ and $\tP\left( \btheta, \cdot \right)$ are measures that admit a similar decomposition, and, in particular for any $\btheta\in\Scal$ and $A\in\Acal$,
\begin{equation}
 P\left( \btheta, A \right) = \int_A Q\left( \btheta, d\btheta' \right) \alpha\left( \btheta \rightarrow \btheta' \right) + \delta_{\btheta}\left( A \right)r\left( \btheta \right),
 \quad \quad
  \tP\left( \btheta, A \right) = \int_A Q\left( \btheta, d\btheta' \right) \tilde{\alpha}\left( \btheta \rightarrow \btheta' \right) + \delta_{\btheta}\left( A \right)\tilde{r}\left( \btheta \right)
\end{equation}
so that the signed measure $\mu_{\btheta} := P\left( \btheta, \cdot \right) - \tP\left( \btheta, \cdot \right)$ decomposes as:
\begin{equation}
 \mu_{\btheta}\left( A \right) = \int_A Q\left( \btheta, d\btheta' \right)\left( \alpha\left( \btheta \rightarrow \btheta' \right) - \tilde{\alpha}\left( \btheta \rightarrow \btheta' \right) \right)  + \delta_{\btheta}\left( A \right)\left( r\left( \btheta \right) - \tilde{r}\left( \btheta \right) \right) \ .
\end{equation}
It is well known that for this type of measure with atoms, the total variation verifies:
\begin{equation}
 \|\mu_{\btheta}\| = \frac{1}{2} \int_A Q\left( \btheta, d\btheta' \right)\left| \alpha\left( \btheta \rightarrow \btheta' \right) - \tilde{\alpha}\left( \btheta \rightarrow \btheta' \right) \right|  + \frac{1}{2}\left| r\left( \btheta \right) - \tilde{r}\left( \btheta \right) \right| \ .
\end{equation}
But since $r\left( \btheta \right) = \int Q\left( \btheta, d\btheta' \right)\left(1 - \alpha\left( \btheta \rightarrow \btheta' \right)\right)$, with a similar result for $\tilde{r}\left( \btheta \right)$, we have that:
\begin{equation}
 \left|r\left( \btheta \right) - \tilde{r}\left( \btheta \right)\right| = \left| \int Q\left( \btheta, d\btheta' \right)\left( \alpha\left( \btheta \rightarrow \btheta' \right) - \tilde{\alpha}\left( \btheta \rightarrow \btheta' \right) \right) \right|
 \leq
 \int Q\left( \btheta, d\btheta' \right)\left| \alpha\left( \btheta \rightarrow \btheta' \right) - \tilde{\alpha}\left( \btheta \rightarrow \btheta' \right) \right| \ .
\end{equation}
Since by assumption $\left| \alpha\left( \btheta \rightarrow \btheta' \right) - \tilde{\alpha}\left( \btheta \rightarrow \btheta' \right) \right|\leq \omega$, we have that:
\begin{equation}
 \| \mu_{\btheta}\| \leq \int Q\left( \btheta, d\btheta' \right)\left| \alpha\left( \btheta \rightarrow \btheta' \right) - \tilde{\alpha}\left( \btheta \rightarrow \btheta' \right)\right| \leq  \omega \ ,
\end{equation}
which is independent of $\btheta$ and thus concludes the proof.
\end{proof}

\subsection{Proof of Proposition \ref{prop_subadditive}}\label{app:prop_subadditive}
\begin{proof}
 If $R=1$ then $\|P_{[1]} - \tilde{P}_{[1]}\| = \|P_{1} - \tilde{P}_{1}\|$.
 If $R=2$, then, using the fact that the signed kernel $P_r - \tilde{P}_r$ satisfies conditions of Lemma \ref{lem2}, we have
 \begin{equation}
  \begin{split}
 \|P_{[2]} - \tilde{P}_{[2]}\| &= \|P_1P_2 - \tilde{P}_1\tilde{P}_2\| \\
 &= \|P_1(P_2-\tilde{P}_2) + \tilde{P}_2(P_1 - \tilde{P}_1)\| \\
 &\leq \|P_1(P_2-\tilde{P}_2)\| + \|\tilde{P}_2(P_1 - \tilde{P}_1)\| \\
 &\leq \|P_1\|\|P_2-\tilde{P}_2\| + \|\tilde{P}_2\|\|P_1 - \tilde{P}_1\| \\
 &= \|P_2 - \tilde{P}_2\| + \|P_1 - \tilde{P}_1\|\,,
  \end{split}
 \end{equation}
where last inequality follows from Lemma \ref{lem1}. Now, we assume that \eqref{eq:thm_scan_1} is valid for every $r \leq R-1$, and prove the statement for $r=R$.
We note that $P_{[R-1]}$ and $\tilde{P}_{[R-1]}$ are both Markov kernels. Hence:
 \begin{equation}
  \begin{split}
 \|P_{[R]} - \tilde{P}_{[R]}\| &= \|P_RP_{[R-1]} - \tilde{P}_R\tilde{P}_{[R-1]}\| \\
 &\leq \|P_R(P_{[R-1]}-\tilde{P}_{[R-1]})\| + \|\tilde{P}_{[R-1]}(P_R - \tilde{P}_R)\| \\
 &= \|P_R - \tilde{P}_R\| + \|P_{[R-1]} - \tilde{P}_{[R-1]}\|
  \end{split}
 \end{equation}
proving the proposition by mathematical induction.
\end{proof}

\subsection{Preliminary results}
We define here two independent results which are needed to prove the main theorem of our paper.

\begin{lemma}\label{app:lemma_1}
Let $x,\ y \in \mathbb{R}^+$, then:
\begin{equation}
|1\wedge x - 1\wedge y| \leq |\log x - \log y| \ .
\end{equation}
\end{lemma}
\begin{proof}
 We note that:
 \begin{equation}
  |1\wedge x - 1\wedge y| \leq |\log (1\wedge x) - \log (1\wedge y)| \ ,
 \end{equation}
because the logarithm function acts as a expansive mapping locally in the set $[0,1]$, or, equivalently, the exponential function is a contraction in the set $[0,1]$.

Then, we can proceed with:
 \begin{equation}
  |\log (1\wedge x) - \log (1\wedge y)| \leq |1\wedge \log x - 1\wedge \log y| \ .
 \end{equation}

The following cases are now possible:
\begin{enumerate}
 \item If $x\geq 1$ and $y \geq 1$, then $|1\wedge \log x - 1\wedge \log y| \leq 0$.
 \item If $x < 1$ and $y < 1$, then $|1\wedge \log x - 1\wedge \log y| = |\log x - \log y|$.
 \item If $x < 1$ and $y \geq 1$ (or, analogously, $x \geq 1$ and $y < 1$), then
 $$
 |1\wedge \log x - 1\wedge \log y| \leq |\log x - 0| \leq |\log x - \log y| \ ,
 $$
 because $\log x$ is a negative number.
\end{enumerate}
 As a worst case scenario across the above cases, we have the required statement
 $$
 |1\wedge x - 1\wedge y| \leq |\log x - \log y| \ .
 $$
\end{proof}

\begin{lemma}\label{app:lemma_2}
 Let $x$, $y$ and $z$ be real numbers such that $x \wedge y \geq z > 0$. Then:
 \begin{equation}
  |\log x - \log y| \leq \frac{1}{z}|x - y| \ .
 \end{equation}
\begin{proof}
 We assume, without loss of generality, that $x \geq y$. We have:
 \begin{equation}
  |\log x - \log y|
  = \left|\log \frac{x}{y}\right|
  = \left|\log \frac{y+x-y}{y}\right|
  = \left|\log \left(1 + \frac{x-y}{y}\right)\right|
  \leq \left|\frac{x-y}{y}\right|
  \leq \frac{1}{y}|x-y| \ ;
 \end{equation}
 where we have used the fact that $\log (1+u) \leq u$ for any $u \geq 0$.

 In the case of $y \geq x$, the inequality becomes
 \begin{equation}
  |\log x - \log y|
  \leq \frac{1}{x}|x-y| \ ;
 \end{equation}
so, we can combine the two inequalities to obtain the statement of the lemma, for any $z \leq x \wedge y$.
\end{proof}
\end{lemma}

\begin{proposition}\label{app:proposition_1}
 If the edge probability function $\rho\left( d, \boldsymbol{\psi} \right)$ defined in Assumption \ref{assumption_edge_probability} satisfies:
 $$
 \left|\rho\left( d, \boldsymbol{\psi} \right) - \rho\left( \tilde{d}, \boldsymbol{\psi} \right)\right| \leq \kappa_0|d-\tilde{d}| \ ,
 $$
 for some distances $d$ and $\tilde{d}$, model parameters $\boldsymbol{\psi}$, and positive constant $\kappa_0$, then it also satisfies:
 $$
 \left|\log \rho\left( d, \boldsymbol{\psi} \right) - \log \rho\left( \tilde{d}, \boldsymbol{\psi} \right)\right| \leq \kappa_1|d-\tilde{d}| \ ;
 $$
 for a suitable positive constant $\kappa_1$.
\end{proposition}
\begin{proof}
Thanks to Lemma \ref{app:lemma_2}, we have:
 $$
 \left|\log \rho\left( d, \boldsymbol{\psi} \right) - \log \rho\left( \tilde{d}, \boldsymbol{\psi} \right)\right|
 \leq \frac{1}{\rho\left( d, \boldsymbol{\psi} \right) \wedge \rho\left( \tilde{d}, \boldsymbol{\psi} \right)}
 \left|\rho\left( d, \boldsymbol{\psi} \right) - \rho\left( \tilde{d}, \boldsymbol{\psi} \right)\right|
 \leq \frac{1}{\rho^\mathcal{L}}
 \left|\rho\left( d, \boldsymbol{\psi} \right) - \rho\left( \tilde{d}, \boldsymbol{\psi} \right)\right| \ ;
 $$
 where $\rho^\mathcal{L} = \inf\limits_{d\geq 0}\left\{ \rho\left( d, \boldsymbol{\psi} \right) \right\}$.
We can now use the Proposition's assumption to obtain the result.
\end{proof}

\begin{corollary}\label{app:corollary_1}
 Analogously to the previous proposition, we have that
 $$
 \left|\log \left[1-\rho\left( d, \boldsymbol{\psi} \right)\right] - \log \left[1-\rho\left( \tilde{d}, \boldsymbol{\psi} \right)\right]\right| \leq \kappa_2|d - \tilde{d}| \ ,
 $$
 for a suitable positive constant $\kappa_2$.
\end{corollary}

\subsection{Proof of Theorem \ref{theorem_acceptance_probs_1}}\label{app:theorem_acceptance_probs_1}
\begin{proof}
 First we use Lemma \ref{app:lemma_1} to obtain:
 \begin{equation}\label{theorem_acceptance_probs_1_eq_1}
  \left| \alpha_{\mathcal{Z}}\left( \textbf{z}_i\rightarrow \textbf{z}_i'\right) - \tilde{\alpha}_{\mathcal{Z}}\left( \textbf{z}_i\rightarrow \textbf{z}_i'\right)\right|
  \leq \left|
  \log \frac{\pi\left( \textbf{z}_i' \middle \vert \mathcal{Z}_{-i},\bpsi,\mathcal{Y} \right)}
  {\pi\left( \textbf{z}_i \middle \vert \mathcal{Z}_{-i},\bpsi,\mathcal{Y} \right)}
  - \log \frac{\tilde{\pi}\left( \textbf{z}_i' \middle \vert \mathcal{Z}_{-i},\bpsi,\mathcal{Y} \right)}
  {\tilde{\pi}\left( \textbf{z}_i \middle \vert \mathcal{Z}_{-i},\bpsi,\mathcal{Y} \right)}
  \right| \ .
 \end{equation}

 Simplify the notation with:
 \begin{align*}
 &p_{ij} = p\left( \textbf{z}_i, \textbf{z}_j; \boldsymbol{\psi} \right) \\
 &p_{ij}' = p\left( \textbf{z}_i', \textbf{z}_j; \boldsymbol{\psi} \right) \\
 &q_{ij} = 1 - p\left( \textbf{z}_i, \textbf{z}_j; \boldsymbol{\psi} \right) \\
 &q_{ij}' = 1 - p\left( \textbf{z}_i', \textbf{z}_j; \boldsymbol{\psi} \right) \\
\end{align*}
with approximate counterparts indicated with a tilde, respectively.

 We note that in the right hand side of \eqref{theorem_acceptance_probs_1_eq_1}, all the terms referring to the prior and proposal distributions cancel each other out. The only remaning terms are:

\begin{equation}
\begin{split}
&\left| \alpha_{\mathcal{Z}}\left( \textbf{z}_i\rightarrow \textbf{z}_i'\right) - \tilde{\alpha}_{\mathcal{Z}}\left( \textbf{z}_i\rightarrow \textbf{z}_i'\right)\right|
\leq  \\
&\hspace{2cm}\leq\left|\sum_{j\in\mathcal{V}\backslash\{i\}} y_{ij}\left[ \log p_{ij}' - \log \tilde{p}_{ij}' - \log p_{ij} + \log \tilde{p}_{ij}\right] + \right. \\
&\hspace{4cm}+ \left. \sum_{j\in\mathcal{V}\backslash\{i\}} (1-y_{ij})\left[ \log q_{ij}' - \log \tilde{q}_{ij}' - \log q_{ij} + \log \tilde{q}_{ij} \right] \right| \\
&\hspace{2cm}\leq\sum_{j\in\mathcal{V}\backslash\{i\}} y_{ij}\left[ \left|\log p_{ij}' - \log \tilde{p}_{ij}'\right| + \left|\log p_{ij} - \log \tilde{p}_{ij}\right|\right] +   \\
&\hspace{4cm}+  \sum_{j\in\mathcal{V}\backslash\{i\}} (1-y_{ij})\left[\left| \log q_{ij}' - \log \tilde{q}_{ij}' \right| + \left|\log q_{ij} - \log \tilde{q}_{ij} \right|\right] \ .
\end{split}
\end{equation}

We note that each of the terms in absolute values are upper bounded by $\kappa b$, for some positive constant $\kappa$. This follows from Proposition \ref{app:proposition_1}, whereby $|d - \tilde{d}| \leq b\sqrt{2}$ by construction. Hence, we can conclude with
\begin{equation}
 \left| \alpha_{\mathcal{Z}}\left( \textbf{z}_i\rightarrow \textbf{z}_i'\right) - \tilde{\alpha}_{\mathcal{Z}}\left( \textbf{z}_i\rightarrow \textbf{z}_i'\right)\right| \leq \kappa' b N \ .
\end{equation}
The proof for the global parameters is done analogously.

\end{proof}

\subsection{Uniform convergence of Metropolis-within-Gibbs kernels operating on a compact state space}\label{app:thm_uniform}

\begin{theorem}\label{thm_uniform}
Let $\Scal$ be a bounded state space with $\Scal\subset \mathbb{R}^d$ (for some $d>0$) and $\mathcal{A}$ be a sigma-algebra on $\Scal$. Let $P$ be a Gibbs kernel operating on $\Scal\times\Acal$ with invariant distribution $\pi$ defined on $(\Scal,\Acal)$. Then the function $u\mapsto \|P(u,,\cdot)^t-\pi\|$ converges uniformly to $0$ as $t\to\infty$, at a geometric rate.
\end{theorem}
\def\ualpha{\underline{\alpha}}
\begin{proof}
For simplicity, we take the case $d=3$, but generalizing the following reasoning for all $d>0$ is straightforward.
Denoting with $P_i$ the MwG kernel that keeps $x_{-i}:=(x_1,\ldots,x_{i-1},x_{i+1},\ldots,x_d)$ fixed, we have for all $x\in\Scal$
\begin{equation}
\label{eq:unif1}
P_i(x,\rmd x')=\left\{Q_i(x,\rmd x_i')\alpha_i(x,x')+\delta_{x_i}(\rmd x'_i)\rho_i(x)\right\}\delta_{x_{-i}}(\rmd x'_{-i})\,,
\end{equation}
where $Q_i$ is the proposal kernel of the $i$-th dimension, $\alpha_i(x,x')=1\wedge \pi(x')Q_i(x'x)\slash \pi(x)Q_i(x,x')$ and $\rho_i(x)=1-\int Q(x,\rmd x')\alpha_i(x,x')$. With regulatory conditions on the proposal kernels $Q_1,Q_2,\ldots$ and since the state space is compact, we have for all $i\in\{1,\ldots,d\}$:
\begin{equation}
\label{eq:unif2}
\bQ_i:=\sup_{(x,y)\in\Scal^2}Q_i(x,y)<\infty\,,\qquad \uQ_i:=\inf_{(x,y)\in\Scal^2}Q_i(x,y)>0\,.
\end{equation}
Moreover, since the pdf of $\pi$ is a continuous function and $\Scal$ is bounded, we have:
\begin{equation}
\label{eq:unif3}
0<\upi\leq \pi(x)\leq\bpi<\infty\,.
\end{equation}
Assuming that, for all $i$, $Q_i$ is absolutely dominated by a common dominating measure, we have that $Q_i(x,\rmd x_i')=Q(x,x_i')\rmd x_i'$ which combined with Eqs. \eqref{eq:unif1}, \eqref{eq:unif2} and \eqref{eq:unif3} yields
\begin{equation}
P_i(x,\rmd x')\geq Q_i(x,x_i')\alpha_i(x,x')\delta_{x_{-i}}(\rmd x'_{-i})\rmd x_i'\geq \uQ_i\ualpha_i\delta_{x_{-i}}(\rmd x'_{-i})\rmd x_i'\,,
\end{equation}
where $\ualpha_i:=\upi\uQ_i/\bpi\bQ_i$. Now, the (systematic-scan) Metropolis-within-Gibbs transition kernel writes
\begin{eqnarray*}
P(x,\rmd x'):=P_1P_2P_3(x,\rmd x')\hspace{-0.6cm}&&=\int P_1(x,\rmd y) P_2P_3(y,\rmd x')\,,\\
\hspace{-0.6cm}&&\geq \int \uQ_1 \ualpha_1\rmd y_1P_2P_3(y_1,x_2,x_3,\rmd x')\,,\\
\hspace{-0.6cm}&&\geq \int \uQ_1 \ualpha_1\rmd y_1\int \uQ_2\ualpha_2 \rmd z_2 P_3(y_1,z_2,x_3,\rmd x')\,,\\
\hspace{-0.6cm}&&\geq \int \uQ_1 \ualpha_1\rmd y_1\int \uQ_2\ualpha_2 \rmd z_2 \uQ_3 \left\{1\wedge \frac{\pi(x')\bQ_3}{\upi \uQ_3}\right\}\rmd x_3'\delta_{y_1}(\rmd x'_1)\delta_{z_2}(\rmd x'_2)\,,\\
\hspace{-0.6cm}&&\geq \left\{\prod_{i=1}^2\uQ_i\ualpha_i\right\}\uQ_3\left\{1\wedge \frac{\pi(x')\uQ_3}{\bpi \bQ_3}\right\} \rmd x'\,,
\end{eqnarray*}
since $\iint \rmd y_1\delta_{y_1}(\rmd x_1')\rmd z_2\delta_{z_2}(\rmd x_2')=\rmd x_1'\rmd x_2'$. Hence, defining $\nu$ as the absolutely continuous probability measure with pdf $\nu(x)\propto 1\wedge {\pi(x')\uQ_3}\slash{\bpi \bQ_3}$, we have
\begin{equation}
\label{eq:unif4}
P(x,\rmd x')\geq \beta \nu(\rmd x')\,,
\end{equation}
with $\beta:=\left\{\prod_{i=1}^2\uQ_i\ualpha_i\right\}\uQ_3\int 1\wedge {\pi(x')\uQ_3}\slash{\bpi \bQ_3}\rmd x'$. We conclude from Eq. \eqref{eq:unif4} that the whole state space $\Scal$ is small for $P$ and that therefore $P$ is uniformly ergodic (with geometric rate $1-\beta$), see e.g. Theorem 8 in \textcite{roberts2004general}.
\end{proof}
\end{document}